%% file: main.tex
\documentclass{lmcs} 
\pdfoutput=1

\usepackage[utf8]{inputenc}

\usepackage{lastpage}
\lmcsdoi{18}{3}{36}
\lmcsheading{}{\pageref{LastPage}}{}{}%
{Feb.~28,~2020}{Sep.~21,~2022}{}


\theoremstyle{plain} 

\def\cf{{\em cf.}}

\usepackage{amsthm}
\usepackage{amscd}
\usepackage{amsmath}
\usepackage{amssymb}
\usepackage{mathtools}
\usepackage{extarrows}
\usepackage{arydshln}
\usepackage{url}
\usepackage{subcaption}
\usepackage{tikz-cd}
\usetikzlibrary{decorations.pathreplacing, positioning}

\newcommand{\bbK}{\mathbb K}
\newcommand{\bbN}{\mathbb N}
\newcommand{\bbP}{\mathbb P}

\newcommand{\bbZ}{\mathbb Z}
\newcommand{\cM}{\mathcal M}
\newcommand{\cR}{\mathcal R}

\newcommand{\cZ}{\mathcal Z}
\newcommand{\fR}{\mathfrak R}
\newcommand{\bfP}{\mathbf P}
\newcommand{\rrule}{\rightarrow}
\newcommand{\rewrst}[1]{\xleftrightarrow{*}_{#1}}
\newcommand{\tensor}[3]{#1\otimes_{#2}#3}

\DeclareMathOperator{\im}{im}
\DeclareMathOperator{\rank}{rank}
\DeclareMathOperator{\CR}{CP}
\DeclareMathOperator{\nr}{nr}
\DeclareMathOperator{\defc}{def}
\DeclareMathOperator{\pos}{Pos}

\newcommand{\ave}{{\sf ave}}
\newcommand{\s}{{\sf s}}
\newcommand{\zero}{{\sf 0}}
\newcommand{\facmon}[2]{\cM_{(#1,#2)}}
\newcommand{\freemod}[2]{#1\underline{#2}}

\newcommand{\imonring}[1]{\bbZ\langle#1\rangle}

\makeatletter
\def\quot{\relax\ifmmode\delimiter"502F30E\mathopen{}\else/\fi}
\makeatother

\begin{document}

\title[A Lower Bound ... By Homological Methods]{A Lower Bound of the Number of Rewrite Rules Obtained by Homological Methods}

\author[M.~Ikebuchi]{Mirai Ikebuchi}	
\address{Massachusetts Institute of Technology, Cambridge, MA, USA}	
\email{ikebuchi@mit.edu}  
\thanks{The author thanks Keisuke Nakano for comments that greatly improved the manuscript and Aart Middeldorp for his suggestion about prime critical pairs.}	

\begin{abstract}
  It is well-known that some equational theories such as groups or boolean algebras can be defined by fewer equational axioms than the original axioms.
  However, it is not easy to determine if a given set of axioms is the smallest or not.
  Malbos and Mimram investigated a general method to find a lower bound of the cardinality of the set of equational axioms (or rewrite rules) that is equivalent to a given equational theory (or term rewriting systems), using homological algebra. Their method is an analog of Squier's homology theory on string rewriting systems.
  In this paper, we develop the homology theory for term rewriting systems more and provide a better lower bound under a stronger notion of equivalence than their equivalence.
  The author also implemented a program to compute the lower bounds, and experimented with 64 complete TRSs.
\end{abstract}

\maketitle

\section{Introduction}
\input{introduction.tex}

\section{Main Theorem}
\input{algorithm.tex}

\section{Preliminaries on Algebra}
\input{preliminaries.tex}


\section{An Overview of the Homology Theory of Algebraic Theories}
\input{homology.tex}

\section{Proof of Main Theorem}
\input{mainthm.tex}

\section{Prime Critical Pairs in a Homological Perspective}
\input{prime.tex}

\section{Deficiency and Computability}
\input{impossible.tex}

\section{Conclusions}
\input{concl.tex}



\bibliographystyle{alphaurl}
\bibliography{main.bib}

\appendix

\section{The matrix $D(R)$ for The Theory of Groups}
\input{group.tex}
\newpage
\section{Experimental Results}
\input{result.tex}
\newpage
\input{result_cont.tex}

\end{document}

%% file: introduction.tex
The purpose of this paper is to find a lower bound of the number of axioms that are equivalent to a given equational theory.
For example, the theory of groups is given by the following axioms:
\begin{equation}\label{eqn:group}
  \hspace{-4mm} 
\begin{array}{ll@{}l}
  G_1.\ m(m(x_1,x_2),x_3) = m(x_1,m(x_2,x_3)), & G_2.\  m(x_1,e) = x_1, & G_3\   m(e,x_1)=x_1,\\
  G_4.\ m(i(x_1),x_1) = e, & G_5.\  m(x_1,i(x_1)) = e. &
\end{array}
\end{equation}
It is well-known that $G_2$ and $G_5$ can be derived from only $\{G_1,G_3,G_4\}$.
Moreover, the theory of groups can be given by two axioms: the axiom
\[
m(x_1,i(m(m(i(m(i(x_2) , m(i(x_1) ,x_3))) , x_4) ,i(m(x_2,x_4))))) = x_3
\]
together with $G_4$ is equivalent to the group axioms \cite{n81}.
If we use the new symbol $d$ for division instead of multiplication, a single axiom,
\[
d(x_1,d(d(d(x_1,x_1),x_2),x_3),d(d(d(x_1,x_1),x_1),x_3)) = x_2,
\]
is equivalent to the group axioms \cite{hn52}.
However, no single axiom written in symbols $m,i,e$ is equivalent to the group axioms.
This is stated without proof by Tarski \cite{t68} and published proofs are given by Neumann \cite{n81} and Kunen \cite{k92}.
Malbos and Mimram developed a general method to calculate a lower bound of the number of axioms that are ``Tietze-equivalent'' to a given complete term rewriting system (TRS) \cite[Proposition 23]{mm16}.
We state the definition of Tietze equivalence later (Definition \ref{def:tietze}), but roughly speaking, it is an equivalence between equational theories (or TRSs) $(\Sigma_1,R_1)$, $(\Sigma_2,R_2)$ where signatures $\Sigma_1$ and $\Sigma_2$ are not necessarily equal to each other, while the usual equivalence between TRSs is defined for two TRSs $(\Sigma,R_1)$, $(\Sigma,R_2)$ over the same signature (specifically, by ${\xleftrightarrow{*}_{R_1}} = {\xleftrightarrow{*}_{R_2}}$).
For string rewriting systems (SRSs), a work was given earlier by Squier \cite{s87}, and Malbos and Mimram's work is an extension of Squier's work.
Squier provided a rewriting view for ``homology groups of monoids'', and proved the existence of an SRS that does not have any equivalent SRSs that are finite and complete.

In this paper, we will develop Malbos and Mimram's theory more, and show an inequality which gives a better lower bound of the number of axioms with respect to the usual equivalence between TRSs over the same signature.
For the theory of groups, our inequality gives that the number of axioms equivalent to the group axioms is greater than or equal to 2, so we have another proof of Tarski's theorem above as a special case.
Our lower bound is algorithmically computable if a complete TRS is given.

We will first give the statement of our main theorem and some examples in Section \ref{section:algorithm}.
Then, we will see Malbos-Mimram's work briefly.
The idea of their work is to provide an algebraic structure to TRSs and extract information of the TRSs, called homology groups, which are invariant under Tietze equivalence.
The basics of such algebraic tools are given in Section \ref{section:prelim}.
We will explain how resolutions of modules, a notion from abstract algebra, is related to rewriting systems, which is not written in usual textbooks.
 and we will see the idea of the construction of the homology groups of TRSs in Section \ref{section:homology}.
After that, in Section \ref{section:mainthm}, we will prove our main theorem.
In Section \ref{section:prime}, we show prime critical pairs are enough for our computation at the matrix operation level and also at the abstract algebra level.
In Section \ref{section:impossible}, we study the number of rewrite rules in a different perspective:
the minimum of $\#R-\#\Sigma$ over all equivalent TRSs $(\Sigma,R)$ is called the \emph{deficiency} in group theory and we show that deciding whether the deficiency is less than a given integer or not is computationally impossible.

%% file: algorithm.tex
\label{section:algorithm}
In this section, we will see our main theorem and some examples.
Throughout this paper, we assume that terms are over the set of variables $\{x_1,x_2,\dotsc\}$ and all signatures we consider are unsorted.
For a signature $\Sigma$, let $T(\Sigma)$ denote the set of terms over the signature $\Sigma$ and the set of variables $\{x_1,x_2,\dotsc\}$.
\begin{defi}
  Let $(\Sigma,R)$ be a TRS.
  The degree of $R$, denoted by $\deg(R)$, is defined by
  \[
  \deg(R) = \gcd\{\#_{i}l - \#_{i}r \mid l \rrule r \in R, i=1,2,\dotsc\}
  \]
  where $\#_it$ is the number of occurrences of $x_i$ in $t$ for $t\in T(\Sigma)$ and we define $\gcd\{0\} = 0$ for convenience.
  For example, $\deg(\{f(x_1,x_2,x_2)\rrule x_1,\ g(x_1,x_1,x_1)\rrule e\}) = \gcd\{0,2,3\}=1$.
\end{defi}
Let $(\Sigma,R=\{l_1\rrule r_1,\dotsc,l_n\rrule r_n\})$ be a TRS and  $\CR(R)=\{(t_1,s_1),\dots,(t_m,s_m)\}$ be the set of the critical pairs of $R$.
For any $i\in\{1,\dots,m\}$, let $a_i, b_i$ be the numbers in $\{1,\dots,n\}$ such that the critical pair $(t_i,s_i)$ is obtained by $l_{a_i}\rrule r_{a_i}$ and $l_{b_i}\rrule r_{b_i}$,
that is, $t_i = r_{a_i}\sigma \leftarrow l_{a_i}\sigma = C[l_{b_i}\sigma] \rightarrow C[r_{b_i}\sigma] = s_i$ for some substitution $\sigma$ and single-hole context $C$.
Suppose $R$ is complete.
We fix an arbitrary rewriting strategy and for a term $t$, let $\nr_j(t)$ be the number of times $l_j\rrule r_j$ is used to reduce $t$ into its $R$-normal form with respect to the strategy.
To state our main theorem, we introduce a matrix $D(R)$ and a number $e(R)$:
\begin{defi}
  Suppose $d = \deg(R)$ is prime or $0$.
  If $d=0$, let $\fR$ be $\bbZ$, and if $d$ is prime,
  let $\fR$ be $\bbZ/d\bbZ$ (integers modulo $d$).
  For $1 \le i \le m$, $1 \le j \le n$, let $D(R)_{ij}$ be the integer $\nr_j(s_i) - \nr_j(t_i) + \delta(b_i,j) - \delta(a_i,j)$ where $\delta(x,y)$ is the Kronecker delta.
  The matrix $D(R)$ is defined by $D(R)=(D(R)_{ji})_{j=1,\dots,n,i=1,\dots,m}$.
\end{defi}
\begin{defi}
  Let $\fR$ be $\bbZ$ or $\bbZ/p\bbZ$ for any prime $p$.
  If an $m\times n$ matrix $M$ over $\fR$ is of the form
  \[ \left(
  \begin{matrix}
    e_1 & 0 & \dots & \dots & \dots & \dots & \dots & 0\\
    0 & e_2 & 0 & \dots & \dots & \dots & \dots & 0\\
    \vdots & 0 & \ddots & 0 & \dots & \dots & \dots & \vdots\\
    \vdots & \vdots & 0 & e_r & 0 & \dots & \dots & \vdots\\
    \vdots & \vdots & \vdots & 0 & 0 & \dots & \dots & \vdots\\
    \vdots & \vdots & \vdots & \vdots & \vdots & \ddots & \dots & \vdots\\
    0 & 0 & \dots & \dots & \dots & \dots & \dots & 0
  \end{matrix}
  \right)
  \]
  and $e_i$ divides $e_{i+1}$ for every $1 \le i < r$, we say $M$ is in \emph{Smith normal form}.
  We call $e_i$s the \emph{elementary divisors}.
\end{defi}
It is known that every matrix over $\fR$ can be transformed into Smith normal form by elementary row/column operations, that is, (1) switching a row/column with another row/column, (2) multiplying each entry in a row/column by an invertible element in $\fR$, and (3) adding a multiple of a row/column to another row/column \cite[9.4]{r10}.
(If $d = 0$, the invertible elements in $\fR\cong \bbZ$ are $1$ and $-1$, and if $d$ is prime, any nonzero element in $\fR=\bbZ/d\bbZ$ is invertible.
So, $e(R)$ is equal to the rank of $D(R)$ if $d$ is prime.)
In general, the same fact holds for any principal ideal domain $\fR$.
\begin{defi}
We define $e(R)$ as the number of invertible elements in the Smith normal form of the matrix $D(R)$ over $\fR$.
\end{defi}
We state the main theorem.
\begin{thm} \label{thm:main_explicit}
  Let $(\Sigma,R)$ be a complete TRS and suppose $d=\deg(R)$ is 0 or prime.
  For any set of rules $R'$ equivalent to $R$, i.e., $\xleftrightarrow{*}_{R'} \,=\,\xleftrightarrow{*}_{R}$,
  we have
  \begin{equation} \label{eqn:main_ineq}
  \#R' \ge \#R - e(R).
  \end{equation}
\end{thm}
We shall see some examples.
\begin{exa}
Consider the signature $\Sigma = \{\zero^{(0)}, \s^{(1)}, \ave^{(2)}\}$ and the set $R$ of rules
\[
\begin{array}{lll}
  A_1.\ave(\zero,\zero) \rrule \zero, & A_2.\ave(x_1,\s(x_2)) \rrule \ave(\s(x_1),x_2), & A_3.\ave(\s(\zero),\zero) \rrule \zero,\\
  A_4.\ave(\s(\s(\zero)),\zero) \rrule s(\zero), &A_5.\ave(\s(\s(\s(x_1))),x_2)\rrule \s(\ave(\s(x_1),x_2)).&
\end{array}
\]
$R$ satisfies $\deg(R)=0$ and has one critical pair $C$:
\begin{center}
\includegraphics[scale=0.5]{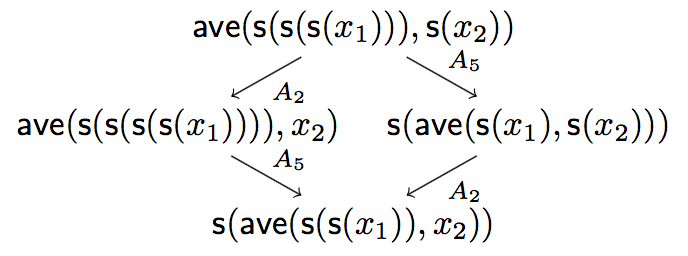}
\end{center}
We can see the matrix $D(R)$ is the $5\times 1$ zero matrix. The zero matrix is already in Smith normal form and $e(R) = 0$.
Thus, for any $R'$ equivalent to $R$, $\#R' \ge \#R = 5$. This means there is no smaller TRS equivalent to $R$.
Also, Malbos-Mimram's lower bound, denoted by $s(H_2(\Sigma,R))$, is equal to 3, though we do not explain how to compute it in this paper.
(We will roughly describe $s(H_2(\Sigma,R))$ in Section \ref{section:homology}.)
\end{exa}
As a generalization of this example, we have an interesting corollary of our main theorem:
\begin{cor}
  Let $(\Sigma,R)$ be a complete TRS.
  If for any critical pair $u \leftarrow t \rightarrow v$, two rewriting paths $t \rightarrow u \rightarrow \dots \rightarrow \hat t$ and $t \rightarrow v \rightarrow \dots \rightarrow \hat t$ contain the same number of $l\rrule r$ for each $l\rrule r\in R$, 
  then there is no $R'$ equivalent to $R$ which satisfies $\#R' < \#R$.
\end{cor}
\begin{exa} \label{example:group}
  We compute the lower bound for the theory of groups, (\ref{eqn:group}).
  A complete TRS $R$ for the theory of groups is given by
  \[
  \begin{array}{ll}
    G_1.\  m(m(x_1,x_2),x_3) \rrule m(x_1,m(x_2,x_3))
    & G_2.\  m(e,x_1) \rrule x_1\\
    G_3.\  m(x_1,e)\rrule x_1
    & G_4.\  m(x_1,i(x_1)) \rrule e\\
    G_5.\  m(i(x_1),x_1)\rrule e
    & G_6.\  m(i(x_1),m(x_1,x_2))\rrule x_2\\
    G_7.\  i(e)\rrule e
    & G_8.\  i(i(x_1))\rrule x_1\\
    G_9.\  m(x_1,m(i(x_1),x_2))\rrule x_2
    & G_{10}.\  i(m(x_1,x_2))\rrule m(i(x_2),i(x_1)).
  \end{array}
  \]
  Since $\deg(R)=2$, we set $\fR = \bbZ/2\bbZ$. $R$ has 48 critical pairs and we get the $10\times 48$ matrix $D(R)$ given in Appendix A.
  The author implemented a program which takes a complete TRS as input and computes its critical pairs, the matrix $D(R)$, and $e(R)$.
  The program is available at \url{https://github.com/mir-ikbch/homtrs}.
  The author checked $e(R) = \rank(D(R)) = 8$ by the program, and also by MATLAB's \texttt{gfrank} function (\url{https://www.mathworks.com/help/comm/ref/gfrank.html}).
  Therefore we have $\#R-e(R)= 2$.
  This provides a new proof that there is no single axiom equivalent to the theory of groups.

  Malbos-Mimram's lower bound is given by $s(H_2(\Sigma,R)) = 0$.
\end{exa}
\begin{exa}
  Let $\Sigma=\{-^{(1)},f^{(1)},+^{(2)},\cdot^{(2)}\}$ and $R$ be
  \[
  \begin{array}{ll}
  A_1.\ -(- x_1)\rrule x_1,& A_2.\ -f(x_1)\rrule f(-x_1),\\
  A_3.\ -(x_1 + x_2) \rrule (-x_1)\cdot(-x_2),& A_4.\ -(x_1\cdot x_2) \rrule (-x_1)+(-x_2).
\end{array}
  \]
  \begin{figure} \label{figure:cps} 
    \includegraphics[scale=0.5]{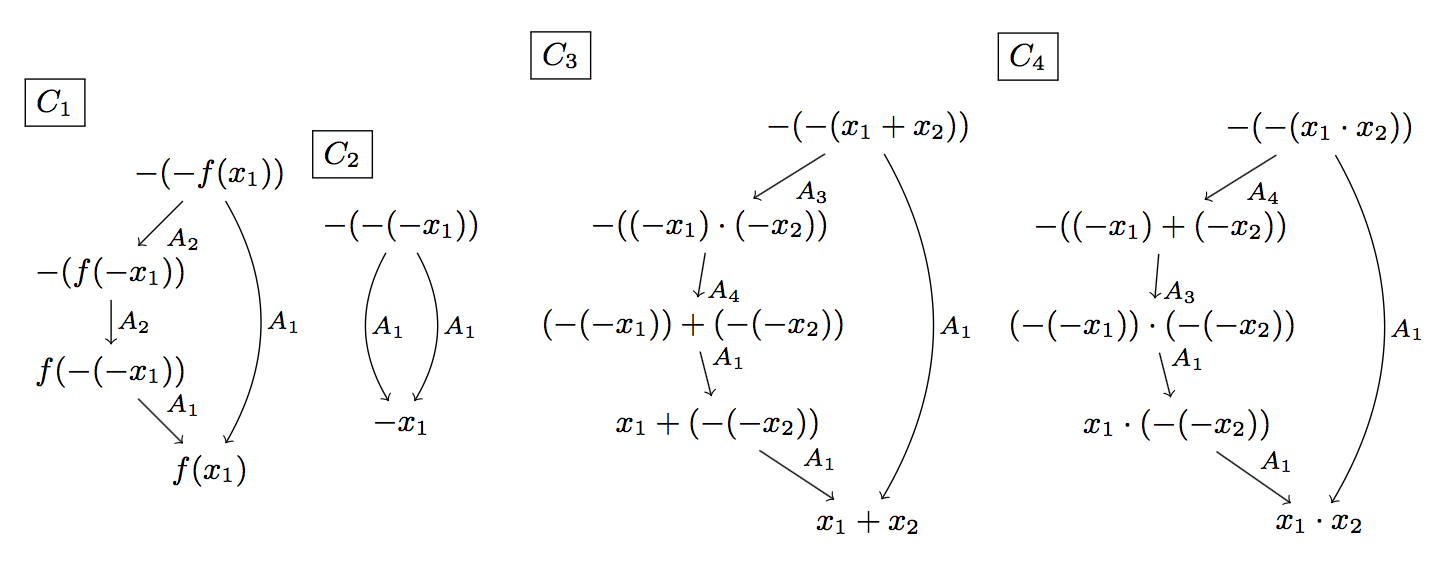}\caption{The critical pairs of $R$}
  \end{figure}
  We have $\deg(R)=0$ and $R$ has four critical pairs (Figure 1).
  The corresponding matrix $D(R)$ and its Smith normal form are computed as
  \[
  D(R)=
  \left(
  \begin{matrix}
    0 & 0 & 1 & 1\\
    2 & 0 & 0 & 0\\
    0 & 0 & 1 & 1\\
    0 & 0 & 1 & 1\\
  \end{matrix}
  \right)
  \rightsquigarrow
    \left(
    \begin{matrix}
      0 & 0 & 1 & 1\\
      2 & 0 & 0 & 0\\
      0 & 0 & 0 & 0\\
      0 & 0 & 0 & 0\\
    \end{matrix}
    \right)
    \rightsquigarrow
    \left(
    \begin{matrix}
      0 & 0 & 1 & 0\\
      2 & 0 & 0 & 0\\
      0 & 0 & 0 & 0\\
      0 & 0 & 0 & 0\\
    \end{matrix}
    \right)
    \rightsquigarrow
    \left(
    \begin{matrix}
      1 & 0 & 0 & 0\\
      0 & 2 & 0 & 0\\
      0 & 0 & 0 & 0\\
      0 & 0 & 0 & 0\\
    \end{matrix}
    \right).
    \]
  Thus, $\#R - e(R) = 3$.
  This tells $R$ does not have any equivalent TRS with 2 or fewer rules, and it is not difficult to see $R$ has an equivalent TRS with 3 rules, $\{A_1,A_2,A_3\}$.

  Malbos-Mimram's lower bound for this TRS is given by $s(H_2(\Sigma,R)) = 1$.
\end{exa}
Although the equality of (\ref{eqn:main_ineq}) is attained for the above three examples, it is not guaranteed the equality is attained by some TRS $R'$ in general.
For example, the TRS with only the associative rule $\{f(f(x_1,x_2),x_3)\rrule f(x_1,f(x_2,x_3))\}$ satisfies $\#R-e(R)=0$ and it is obvious that no TRSs with zero rule is equivalent.
Also, in Appendix B, Malbos-Mimram's and our lower bounds for various examples are given.

%% file: preliminaries.tex
\label{section:prelim}
In this section, we give a brief introduction to module theory, homological algebra, and Squier's theory of homological algebra for string rewriting systems (SRSs) \cite{s87}.
Even though Squier's theory is not directly needed to prove our theorem, it is helpful to understand the homology theory for TRSs, which is more complicated than SRSs' case.

\subsection{Modules and Homological Algebra}
We give basic definitions and theorems on module theory and homological algebra without proofs.
For more details, readers are referred to \cite{r10,r09} for example.

Modules are the generalization of vector spaces in which the set of scalars form a ring, not necessarily a field.
\begin{defi}
  Let $\fR$ be a ring and $(M,+)$ be an abelian group.
  For a map $\cdot:\fR\times M \rightarrow M$, $(M,+,\cdot)$ is a \emph{left $\fR$-module} if for all $r,s \in \fR$ and $x,y \in M$, we have
  \[
    r\cdot(x+y) = r\cdot x + r\cdot y,\
    (r+s)\cdot x = r\cdot x + s\cdot x,\
    (rs)\cdot x = r\cdot(s\cdot x)
  \]
  where $rs$ denotes the multiplication of $r$ and $s$ in $\fR$.
  We call the map $\cdot$ \emph{scalar multiplication}.

  For a map $\cdot : M \times \fR \rightarrow M$, $(M,+,\cdot)$ is a \emph{right $\fR$-module} if for any $r,s \in \fR$ and $x,y \in M$,
  \[
    (x+y)\cdot r = x\cdot r + y\cdot r,\
    x\cdot (r+s) = x\cdot r + x\cdot s,\
    x\cdot(sr) = (x\cdot s)\cdot r.
  \]
\end{defi}
If ring $\fR$ is commutative, we do not distinguish between left $\fR$-modules and right $\fR$-modules and simply call them $\fR$-modules.

Linear maps and isomorphisms of modules are also defined in the same way as for vector spaces.
\begin{defi}
  For two left  $\fR$-modules $(M_1,+_1,\cdot_1),(M_2,+_2,\cdot_2)$, a group homomorphism $f : (M_1,+_1) \rightarrow (M_2,+_2)$ is an \emph{$\fR$-linear map} if it satisfies $f(r\cdot_1 x) = r\cdot_2 f(x)$ for any $r\in \fR$ and $x\in M_1$.
  An $\fR$-linear map $f$ is an \emph{isomorphism} if it is bijective, and two modules are called \emph{isomorphic} if there exists an isomorphism between them.
\end{defi}

\begin{exa}
  Any abelian group $(M,+)$ is a $\bbZ$-module under the scalar multiplication $
  n\cdot x = \underbrace{x+\dotsb + x}_n$.
\end{exa}
\begin{exa}
  For any ring $\fR$, the direct product $\fR^n=\underbrace{\fR\times \dotsb \times\fR}_{n}$ forms a left $\fR$-module under the scalar multiplication $r\cdot(r_1,\dotsc,r_n)=(rr_1,\dotsc,rr_n)$.
\end{exa}
\begin{exa}
  Let $\fR$ be a ring and $X$ be a set.
  $\fR \underline X$ denotes the set of formal linear combinations
  \[
    \sum_{x\in X}r_x \underline x\quad  (r_x \in \fR)
  \]
  where $r_x = 0$ except for finitely many $x$s.
  The underline is added to emphasize a distinction between $r\in\fR$ and $x\in X$.
  $\fR \underline X$ forms a left $\fR$-module under the addition and the scalar multiplication defined by
  \[
  \left(\sum_{x\in X}r_x \underline x\right) + \left(\sum_{x\in X}s_x \underline x\right)
  = \sum_{x \in X}(r_x + s_x) \underline x,\quad
  s\cdot\left(\sum_{x\in X}r_x \underline x\right) = \sum_{x\in X}(sr_x) \underline x.
  \]
  If $X$ is the empty set, $\fR\underline X$ is the left $\fR$-module $\{0\}$ consisting of only the identity element.
  We simply write $0$ for $\{0\}$.
  $\fR \underline X$ is called the \emph{free left $\fR$-module generated by $X$}.
  If $\#X = n\in \bbN$, $\fR \underline X$ can be identified with $\fR^n$.

  A left $\fR$-module $M$ is said to be \emph{free} if $M$ is isomorphic to $\fR\underline{X}$ for some $X$.
  Free modules have some similar properties to vector spaces.
  If a left $\fR$-module $F$ is free, then there exists a basis (i.e., a subset that is linearly independent and generating) of $F$.
  If a free left $\fR$-module $F$ has a basis $(v_1,\dotsc,v_n)$, any $\fR$-linear map $f : F \rightarrow M$ is uniquely determined if the values $f(v_1),\dotsc,f(v_n)$ are specified.
  Suppose $F_1$, $F_2$ are free left $\fR$-modules and $f: F_1 \rightarrow F_2$ is an $\fR$-linear map.
  If $F_1$ has a basis $(v_1,\dotsc,v_n)$ and $F_2$ has a basis $(w_1,\dotsc,w_m)$,
  the matrix $(a_{ij})_{i=1,\dotsc,n,j=1,\dotsc,m}$ where
  $a_{ij}$s satisfy $f(v_i)=a_{i1}w_1 + \dotsb + a_{im}w_m$  for any $i=1,\dotsc,n$
  is called a \emph{matrix representation} of $f$.
\end{exa}
We define submodules and quotient modules, as in linear algebra.
\begin{defi}
  Let $(M,+,\cdot)$ be a left (resp. right) $\fR$-module.
  A subgroup $N$ of $(M,+)$ is a \emph{submodule} if for any $x \in N$ and $r\in\fR$, the scalar multiplication $r\cdot x$ (resp. $x\cdot r$) is in $N$.

  For any submodule $N$, the quotient group $M/N$ is also an $\fR$-module.
  $M/N$ is called the \emph{quotient module} of $M$ by $N$.
\end{defi}
For submodules and quotient modules, the following basic theorem is known:
\begin{thm}[First isomorphism theorem {\cite[Theorem 8.8]{r10}}]
  Let $(M,+,\cdot),(M',+',\cdot')$ be left (or right) $\fR$-modules, and $f : M \rightarrow M'$ be an $\fR$-linear map.
  \begin{enumerate}
    \item The inverse image of $0$ by $f$, $\ker f = \{x \in M\mid f(x) = 0\}$, is a submodule of $M$.
    \item The image of $M$ by $f$, $\im f = \{f(x)\mid x \in M\}$, is a submodule of $M'$.
    \item The image $\im f$ is isomorphic to $M/\ker f$.
  \end{enumerate}
\end{thm}
\begin{thm}[Third isomorphism theorem {\cite[Theorem 7.10]{r10}}]
  Let $M$ be a left (or right) $\fR$-module, $N$ be a submodule of $M$, and $L$ be a submodule of $N$. Then $(M/L)/(N/L)$ is isomorphic to $M/N$.
\end{thm}
\begin{thmC}[{\cite[Theorem 9.8]{r10}}] \label{thm:subfree}
  Let $\fR$ be $\bbZ$ or $\bbZ/p\bbZ$ for some prime $p$.
  Every submodule of a free $\fR$-module is free.
  Moreover, if an $\fR$-module $M$ is isomorphic to $\fR^n$, then every submodule $N$ of $M$ is isomorphic to $\fR^m$ for some $m \le n$.
  (In general, this holds for any principal ideal domain $\fR$.)
\end{thmC}

Let $M$ be a left $\fR$-module.
For $S \subset M$, the set $\fR S$ of all elements in $M$ of the form $\sum_{i=1}^kr_is_i$ $(k\in\bbZ_{\ge 0},r_i \in\fR, s_i \in S)$ is a submodule of $M$.
If  $\fR S = M$, $S$ is called a \emph{generating set} of $S$ and the elements of $S$ are called \emph{generators} of $M$.
Let $S = \{s_i\}_{i\in I}$ be a generating set of $M$ for some indexing set $I$.
For a set $X = \{x_i\}_{i\in I}$, the linear map $\epsilon : \freemod{\fR}{X} \ni x_i \mapsto s_i \in M$ is a surjection from the free module $\freemod{\fR}{X}$.
The elements of $\ker \epsilon$, that is, elements $\sum_{x_i\in X}r_i\underline{x_i}$ satisfying $\epsilon(\sum_{x_i\in X}r_i\underline{x_i}) = \sum_{x_i\in X}r_i s_i = 0$, are called \emph{relations} of $M$.

Now, we introduce one of the most important notions to develop the homology theory of rewriting systems, \textit{free resolutions}.
We first start from the following example.
\begin{exa} \label{example:resolution}
Let $M$ be the $\bbZ$-module defined by
\[
\freemod{\bbZ}{\{a,b,c,d,e\}}/\bbZ\{\underline a+\underline b+\underline c-\underline d-\underline e,\ 2\underline b-\underline c,\ \underline a+2\underline c-\underline b-\underline d-\underline e\}.
\]
We consider the $\bbZ$-linear map between free $\bbZ$-modules $f_0 : \bbZ^3 \rightarrow \freemod{\bbZ}{\{a,b,c,d,e\}}$ defined by
\[
f_0(1,0,0) = \underline a+\underline b+\underline c-\underline d-\underline e,\
f_0(0,1,0) = 2\underline b - \underline c,\
f_0(0,0,1) = \underline a + 2\underline c-\underline b-\underline d-\underline e.
\]
We can see that the image of $f_0$ is the set of relations of $M$.
In other words, $\im f_0 = \ker \epsilon$ for the linear map $\epsilon : \freemod{\bbZ}{\{a,b,c,d,e\}}\rightarrow M$ which maps each element to its equivalence class.
Then, we consider the ``relations between relations'', that is,
triples $(n_1,n_2,n_3)$ which satisfy $f_0(n_1,n_2,n_3)=n_1(\underline a+\underline b+\underline c-\underline d-\underline e)+n_2(2\underline b-\underline c)+n_3(\underline a+2\underline c-\underline b-\underline d-\underline e) = 0$, or equivalently, elements of $\ker f_0$.
We can check $\ker f_0 = \{m(-1,1,1)\mid m \in \bbZ \}$.
This fact can be explained in terms of rewriting systems.
If we write relations in the form of rewrite rules
\[
A_1.\ \underline a+\underline b+\underline c\rrule \underline d+\underline e,\ A_2.\ 2\underline b \rrule \underline c,\  A_3.\ \underline a+2\underline c\rrule \underline b+\underline d+\underline e,
\]
we see $\{A_1,A_2,A_3\}$ is a complete rewriting system (over the signature $\{\underline a,\underline b,\underline c,\underline d,\underline e,+\}$) with two joinable critical pairs
\begin{center}
\includegraphics[scale=0.5]{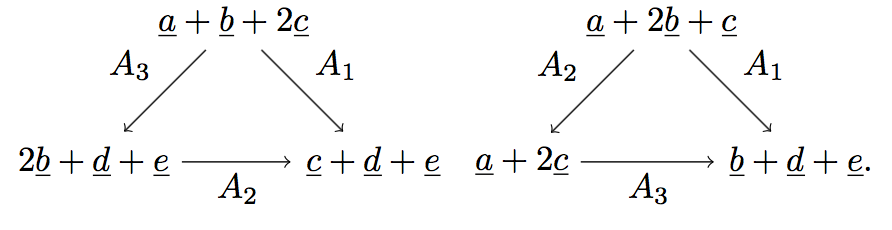}
\end{center}
We associate these critical pairs with an equality between formal sums $A_2 + A_3 = A_1$,
and it corresponds to
\[
  f_0(-1,1,1) = \underbrace{-(\underline a+\underline b+\underline c-\underline d-\underline e)}_{-A_1}+\underbrace{(2\underline b-\underline c)}_{A_2}+\underbrace{(\underline a+2\underline c-\underline b-\underline d-\underline e)}_{A_3} = 0.
\]
In fact, this correspondence between critical pairs and ``relations between relations'' is a key to the homology theory of TRSs.

We define a linear map $f_1 : \bbZ \rightarrow \bbZ^3$ by $f_1(1) = (-1,1,1)$
and then $f_1$ satisfies $\im f_1 = \ker f_0$.
We can go further, that is, we can consider $\ker f_1$, but it clearly turns out that $\ker f_1 = 0$.

We encode the above information in the following diagram:
\begin{equation} \label{eqn:ex_exact}
\bbZ \xrightarrow{f_1} \bbZ^3 \xrightarrow{f_0} \freemod{\bbZ}{\{a,b,c,d,e\}} \xrightarrow{\epsilon} M
\end{equation}
where $\im f_1 = \ker f_0, \im f_0 = \ker \epsilon$ and $\epsilon$ is surjective.
Sequences of modules and linear maps with these conditions are called free resolutions:
\end{exa}
\begin{defi}
  A sequence of left $\fR$-modules and $\fR$-linear maps
  \[
  \dotsb \xrightarrow{f_{i+1}} M_{i+1} \xrightarrow{f_i} M_{i} \xrightarrow{f_{i-1}} \dotsb
  \]
  is called an \emph{exact sequence} if $\im f_{i} = \ker f_{i-1}$ holds for any $i$.

  Let $M$ be a left $\fR$-module. For infinite sequence of free modules $F_i$ and linear maps $f_i : F_{i+1} \rightarrow F_{i}$, $\epsilon : F_0 \rightarrow M$, if the sequence
  \[
  \dotsb\xrightarrow{f_1} F_1 \xrightarrow{f_0} F_0 \xrightarrow{\epsilon} M
  \]
  is exact and $\epsilon$ is surjective, the sequence above is called a \emph{free resolution} of $M$.
  If the sequence is finite, it is called a \emph{partial free resolution}.

  (Exact sequences and free resolutions are defined for right $\fR$-modules in the same way.)
\end{defi}
Notice that the exact sequence (\ref{eqn:ex_exact}) can be extended to the infinite exact sequence
\[
\dotsb\rightarrow 0 \rightarrow \dotsb \rightarrow 0 \rightarrow
\bbZ \xrightarrow{f_1} \bbZ^3 \xrightarrow{f_0} \freemod{\bbZ}{\{a,b,c,d,e\}} \xrightarrow{\epsilon} M
\]
since $\ker f_1 = 0$.
Thus, the sequence (\ref{eqn:ex_exact}) is a free resolution of $M$.

As there are generally several rewriting systems equivalent to a given equational theory, free resolutions of $M$ are not unique.
However, we can construct some information of $M$ from a (partial) free resolution which does not depend on the choice of the free resolution.
The information is called \textit{homology groups}.
To define the homology groups, we introduce the tensor product of modules.
\begin{defi}
  Let $N$ be a right $\fR$-module and $M$ be a left $\fR$-module.
  Let $F(N\times M)$ be the free abelian group generated by $N\times M$.
  The \emph{tensor product} of $N$ and $M$, denoted by $\tensor{N}{\fR}{M}$, is the quotient group of $F(N\times M)$ by the subgroup generated by the elements of the form
  \[
    (x,y) + (x, y') - (x,y + y'),\
    (x,y) + (x',y) - (x+x',y),\
    (x\cdot r,y) - (x,r\cdot y)
  \]
  where $x,x' \in N$, $y,y' \in M$, $r \in R$.
  The equivalence class of $(x,y)$ in $\tensor{N}{\fR}{M}$ is written as $x \otimes y$.

  For a right $\fR$-module $N$ and a $\fR$-linear map $f : M \rightarrow M'$ between left $\fR$-modules $M,M'$, we write $N\otimes f : \tensor{N}{\fR}{M} \rightarrow \tensor{N}{\fR}{M'}$ for the map $(N\otimes f)(a\otimes x) = a \otimes f(x)$.
  $N\otimes f$ is known to be well-defined and be a group homomorphism.
\end{defi}
Let $\dotsb\xrightarrow{f_1}F_1 \xrightarrow{f_0} F_0 \xrightarrow{\epsilon} M$ be a free resolution of a left $\fR$-module $M$.
For a right $\fR$-module $N$, we consider the sequence
\begin{equation} \label{eqn:tensor_free_res}
\dotsb\xrightarrow{N\otimes f_1}\tensor{N}{\fR}{F_1} \xrightarrow{N\otimes f_0} \tensor{N}{\fR}{F_0}.
\end{equation}
Then, it can be shown that $\im (N\otimes f_i) \subset \ker (N\otimes f_{i-1})$ for any $i = 1,2,\dotsc$.
In general, a sequence $\dotsb\xrightarrow{f_{i+1}} M_{i+1} \xrightarrow{f_i} M_{i} \xrightarrow{f_{i-1}} \dotsb$ of left/right $\fR$-modules satisfying $\im f_i \subset \ker f_{i-1}$ for any $i$ is called a \textit{chain complex}.
The homology groups of a chain complex are defined to be the quotient group of $\ker f_{i-1}$ by $\im f_{i}$:
\begin{defi}
  Let $(C_\bullet,f_\bullet)$ denote the pair $(\{C_i\}_{i=0,1,\dots},\{f_i : C_{i+1} \rightarrow C_{i}\}_{i=0,1,\dots})$.
  For a chain complex $\dotsb\xrightarrow{f_{i+1}} C_{i+1} \xrightarrow{f_i} C_{i} \xrightarrow{f_{i-1}} \dotsb$, the abelian group $H_j(C_\bullet,f_\bullet)$ defined by
  \[
  H_j(C_\bullet,f_\bullet) = \ker f_{j-1} / \im f_{j}
  \]
  is called the \emph{$j$-th homology groups} of the chain complex $(C_\bullet,f_\bullet)$.
\end{defi}
The homology groups of the chain complex (\ref{eqn:tensor_free_res}) depend only on $M$, $N$, and $\fR$:
\begin{thmC}[{\cite[Corollary 6.21]{r09}}]\label{thm:hominv}
  Let $M$ be a left $\fR$-module and $N$ be a right $\fR$-module.
  For any two resolutions $\dotsb\xrightarrow{f_1}F_1 \xrightarrow{f_0} F_0 \xrightarrow{\epsilon} M$, $\dotsb\xrightarrow{f'_1}F'_1 \xrightarrow{f'_0} F'_0 \xrightarrow{\epsilon}M$, we have a group isomorphism
  \[
  H_j(\tensor{N}{\fR}{F_\bullet},N\otimes f_\bullet) \cong H_j(\tensor{N}{\fR}{F_\bullet'},N\otimes f_\bullet').
  \]
\end{thmC}
We end this subsection by giving some basic facts on exact sequences.
\begin{propC}[{\cite[Proposition 7.20 and 7.21]{r10}}] \label{thm:exactbasic} \leavevmode
  \begin{enumerate}
    \item $M_1 \xrightarrow{f} M_2 \rightarrow 0$ is exact if and only if $\ker f = 0$.
    \item $0 \rightarrow M_1 \xrightarrow{f} M_2$ is exact if and only if $\im f = M_2$.
    \item If $M_1$ is a submodule of $M_2$, the sequence $0 \rightarrow M_2 \xrightarrow{\iota} M_1 \xrightarrow{\pi} M_1/M_2 \rightarrow 0$ is exact where $\iota$ is the inclusion map $\iota(x) = x$ and $\pi$ is the projection $\pi(x) = [x]$.
  \end{enumerate}
\end{propC}
\begin{prop} \label{thm:split}
  Suppose we have an exact sequence of $\fR$-modules $0 \rightarrow M_1 \rightarrow M_2 \rightarrow M_3 \rightarrow 0$.
  If $M_3$ is free, then $M_2$ is isomorphic to $M_1\times M_3$.
\end{prop}
The proof is given by using \cite[Proposition 7.22]{r10}.

\subsection{String Rewriting Systems and Homology Groups of Monoids}
For an alphabet $\Sigma$, $\Sigma^*$ denotes the set of all strings of symbols over $\Sigma$.
The set $\Sigma^*$ forms a monoid under the operation of concatenation with the empty string serving as the identity, and we call $\Sigma^*$ the free monoid generated by $\Sigma$.
For a string rewriting system (SRS) $(\Sigma,R)$, we write $\facmon{\Sigma}{R}$ for the set defined by $\facmon{\Sigma}{R} = \Sigma^*\quot\rewrst{R}$.
We can see $\facmon{\Sigma}{R}$ is a monoid under the operations $[u]\cdot[v]=[uv]$ where $[w]$ denotes the equivalence class of $w \in \Sigma^*$ with respect to $\rewrst{R}$.

We say that two SRSs $(\Sigma_1,R_1),(\Sigma_2,R_2)$ are \emph{Tietze equivalent} if the monoids $\facmon{\Sigma_1}{R_1}$, $\facmon{\Sigma_2}{R_2}$ are isomorphic.
It is not difficult to show that for any two SRSs $(\Sigma,R_1),(\Sigma,R_2)$ with the same signature, if $R_1$ and $R_2$ are equivalent (i.e., $\rewrst{R_1}\,=\,\rewrst{R_2}$), then $(\Sigma,R_1)$ and $(\Sigma,R_2)$ are isomorphic.
Roughly speaking, the notion that two SRSs are isomorphic means that the SRSs are equivalent but their alphabets can be different.
For example, let $\Sigma_1$ be $\{a,b,c\}$ and $R_1$ be $\{abb\rrule ab,\ ba \rrule c\}$.
Then, $(\Sigma_1,R_1)$ is isomorphic to $(\Sigma_2,R_2)$ where $\Sigma_2 = \{a,b\}$ and $R_2 = \{abb\rrule ab\}$.
Intuitively, since $c$ is equivalent to $ba$ with respect to the congruence $\rewrst{R_1}$, $c$ is redundant as long as we consider strings modulo $\rewrst{R_1}$ and $(\Sigma_2,R_2)$ is the SRS made by removing $c$ from $(\Sigma_1,R_1)$.

If a monoid $S$ is isomorphic to $\facmon{\Sigma}{R}$ for an SRS $(\Sigma,R)$, we call $(\Sigma,R)$ a \emph{presentation} of the monoid $S$.

Let $S$ be a monoid and consider the free $\bbZ$-module $\freemod{\bbZ}{S}$.
The module $\freemod{\bbZ}{S}$ can be equipped with a ring structure under the multiplication
$\left(\sum_{w\in S}n_w\underline w\right)\left(\sum_{w\in S}m_w\underline w\right) = \sum_{w,v\in S}n_wm_v\underline{wv}$
where $n_wm_v$ is the usual multiplication of integers and $wv$ is the multiplication of the monoid $S$.
$\freemod{\bbZ}{S}$ as a ring is called the \emph{integral monoid ring} of $S$.
When we think of $\freemod{\bbZ}{S}$ as a ring, we write $\imonring{S}$ instead of $\freemod{\bbZ}{S}$.

We consider $\imonring{S}$-modules.
The group of integers $\bbZ$ forms a left (resp. right) $\imonring{S}$-module under the scalar multiplication $(\sum_{w\in S}n_w\underline w)\cdot m = \sum_{w\in S}n_w m$ (resp. $m \cdot (\sum_{w\in S}n_w\underline w) = \sum_{w\in S}n_w m \underline w$).
Let $\dotsb\xrightarrow{\partial_1}F_1\xrightarrow{\partial_0}F_0\xrightarrow{\epsilon}\bbZ$
be a free resolution of $\bbZ$ over the ring $\imonring{S}$.
The abelian group $H_i(S)$ is defined as the $i$-th homology group of the chain complex $(\tensor{\bbZ}{\imonring{S}}{F_\bullet},\bbZ\otimes\partial_\bullet)$, i.e.,
\[H_i(S) = H_i(\tensor{\bbZ}{\imonring{S}}{F_\bullet},\bbZ\otimes\partial_\bullet)
= \ker\bbZ\otimes\partial_{i-1}/\im\bbZ\otimes\partial_{i}.
\]
If $S$ is isomorphic to $\facmon{\Sigma}{R}$ for some SRS $(\Sigma,R)$, it is known that there is a free resolution in the form of
\[
\dotsb \rightarrow
\freemod{(\imonring{S})}{P} \xrightarrow{\partial_2} \freemod{(\imonring{S})}{R} \xrightarrow{\partial_1} \freemod{(\imonring{S})}{\Sigma} \xrightarrow{\partial_0} \freemod{(\imonring{S})}{\{\star\}} \xrightarrow{\epsilon} \bbZ
\]
for some set $P$.
Squier \cite{s87} showed that if the SRS $(\Sigma,R)$ is complete and reduced\footnote{An SRS $(\Sigma,R)$ is reduced if for each $l\rrule r \in R$, $r$ is normal w.r.t. $\rightarrow_{R}$ and there does not exist $l'\rrule r'\in R$ such that $l' = ulv \neq l$ for some $u,v\in\Sigma^*$},
 there is $\partial_2 : \freemod{(\imonring{S})}{P} \rightarrow \freemod{(\imonring{S})}{R}$
for $P = (\text{the critical pairs of $R$})$ so that we can compute $H_2(S)=\ker \partial_1/\im\partial_2$ explicitly.
This is an analog of Example \ref{example:resolution}, but we omit the details here.
For an abelian group $G$, let $s(G)$ denote the minimum number of generators of $G$
(i.e., the minimum cardinality of the subset $A \subset G$ such that any element $x \in G$ can be written by $x = a_1 + \dotsb + a_k - a_{k+1} - \dotsb - a_m$ for $a_1,\dotsc,a_m \in A$).
Then, we have the following theorem:
\begin{thmC}[\cite{u95}] \label{thm:srs_ineq}
  Let $(\Sigma,R)$ be an SRS and $S = \facmon{\Sigma}{R}$. Then
  $\#\Sigma \ge s(H_1(S))$, $\#R \ge s(H_2(S))$.
\end{thmC}
To prove this theorem, we use the following lemma:
\begin{lem}
  Let $X$ be a set. The group homomorphism $\tensor{\bbZ}{\imonring{S}}{\freemod{(\imonring{S})}{X}} \rightarrow \freemod{\bbZ}{X}$,
  $n\langle w\rangle\underline x \mapsto n\underline x$ is an isomorphism.
\end{lem}
This lemma is proved in a straightforward way.
\begin{proof}[Proof of Theorem \ref{thm:srs_ineq}]
  Since $\tensor{\bbZ}{\imonring{S}}{(\freemod{\imonring{S})}{X}} \cong \freemod{\bbZ}{X}$ by the above lemma, $s(\tensor{\bbZ}{\imonring{S}}{(\freemod{\imonring{S})}{X}}) = s(\freemod{\bbZ}{X}) = \#X$.
  For any set $Y$ and group homomorphism $f : \freemod{\bbZ}{X} \rightarrow \freemod{\bbZ}{Y}$, since $\ker f$ is a subgroup of $\freemod{\bbZ}{X}$, we have $\#X \ge s(\ker f)$.
  For any subgroup $H$ of $\ker f$, $\ker f/H$ is generated by $[x_1],\dotsc,[x_k]$ if $\ker f$ is generated by $x_1,\dotsc,x_k$.
  Thus $\#\Sigma \ge s(\ker \partial_0/\im\partial_1) = s(H_1(S))$, $\#R \ge s(\ker\partial_1/\im\partial_2) = s(H_2(S))$.
\end{proof}
Note that $H_i(S)$ does not depend on the choice of presentation $(\Sigma,R)$ by Theorem \ref{thm:hominv}.
Therefore, Theorem \ref{thm:srs_ineq} can be restated as follows:
Let $(\Sigma,R)$ be an SRS. For any SRS $(\Sigma',R')$ isomorphic to $(\Sigma,R)$, the number of symbols $\#\Sigma'$ is bounded below by $s(H_1(\facmon{\Sigma}{R}))$ and the number of rules $\#R'$ is bounded below by $s(H_2(\facmon{\Sigma}{R}))$.

%% file: homology.tex
\label{section:homology}
In this section, we will briefly see the homology theory of algebraic theories, which is the main tool to obtain our lower bounds.

We fix a signature $\Sigma$.
Let $t = \langle t_1,\dotsc,t_n\rangle$ be a $n$-tuple of terms and suppose that for each $t_i$, the set of variables in $t_i$ is included in $\left\{x_1,\dots,x_m\right\}$.
For an $m$-tuple of term $s=\langle s_1,\dotsc,s_m\rangle$, we define the composition of $t$ and $s$ by
\[
t\circ s = \langle t_1[s_1/x_1,\dotsc,s_m/x_m],\dotsc,t_n[s_1/x_1,\dotsc,s_m/x_m]\rangle
\]
where $t_i[s_1/x_1,\dotsc,s_m/x_m]$ denotes the term obtained by substituting $s_j$ for $x_j$ in $t_i$ for each $j=1,\dotsc,m$ in parallel.
(For example, $f(x_1,x_2)[g(x_2)/x_1,g(x_1)/x_2] = f(g(x_2),g(x_1))$.)
By this definition, we can think of any $m$-tuple $\langle s_1,\dotsc,s_m\rangle$ of terms as a (parallel) substitution $\left\{x_1 \mapsto s_1,\dotsc,x_m\mapsto s_m\right\}$.
Recall that, for a TRS $R$, the reduction relation $\rightarrow_R$ between terms is defined as
$t_1\rightarrow_R t_2 \iff t_1=C[l\circ s],\ t_2=C[r\circ s]$
for some single-hole context $C$, $m$-tuple $s$ of terms, and rewrite rule $l \rrule r\in R$ whose variables are included in $\{x_1,\dotsc,x_m\}$.
This definition suggests that the pair of a context $C$ and an $m$-tuple of terms (or equivalently, substitution) $s$ is useful to think about rewrite relations.
Malbos and Mimram \cite{mm16} called the pair of a context and an $m$-tuple of terms a \textit{bicontext}.
For a bicontext $(C,t)$ and a rewrite rule $A$, we call the triple $(C,A,t)$ a \textit{rewriting step}.
The pair of two rewriting steps $(\square,l_1\rrule r_1,s),(C,l_2\rrule r_2,t)$ is called a \textit{critical pair} if the pair $(r_1\circ s, C[r_2\circ t])$ of terms is a critical pair in the usual sense given by $l_1 \rightarrow r_1$, $l_2\rightarrow r_2$.

The composition of two bicontexts $(C,t),(D,s)$ ($t=\langle t_1,\dotsc,t_n\rangle$, $s=\langle s_1,\dotsc,s_m\rangle$) is defined by
\[
(C,t)\circ (D,s)=(C[D\circ t],s\circ t)
\]
where $D \circ t = D[t_1/x_1,\dotsc,t_n/x_n]$ and note that the order of composition is reversed in the second component.
We write $\bbK(n,m)$ ($n,m\in \bbN$) for the set of bicontexts $(C,t)$ where $t=\langle t_1,\dotsc,t_n\rangle$ and each $t_i$ and $C$ have variables in $\{x_1,\dotsc,x_m\}$ (except $\square$ in $C$).

To apply homological algebra to TRSs, we construct an algebraic structure from bicontexts.
For two natural numbers $n,m$, we define $\imonring{\bbK}(n,m)$ to be the free abelian group generated by $\bbK(n,m)$
(i.e., any element in $\imonring{\bbK}(n,m)$ is written in the form of formal sum $\sum_{(C,t)\in\bbK(n,m)}\lambda_{(C,t)}(C,t)$ where each $\lambda_{(C,t)}$ is in $\bbZ$ and is equal to $0$ except for finitely many $(C,t)$s.)
Then, the composition $\circ : \bbK(n,m)\times \bbK(k,n)\rightarrow \bbK(k,m)$ can be extended to $\circ : \imonring{\bbK}(n,m)\times \imonring{\bbK}(k,n)\rightarrow \imonring{\bbK}(k,m)$ by
\[
\left(\sum_{(C,t)}\lambda_{(C,t)}(C,t)\right)\circ
\left(\sum_{(D,s)}\mu_{(D,s)}(D,s)\right)
= \sum_{(C,t)}\sum_{(D,s)}\lambda_{(C,t)}\mu_{(D,s)}((C,t)\circ(D,s)).
\]
This family of free abelian groups forms a structure called \emph{ringoid}.
%
\begin{defi}
  Suppose an abelian group $(\cR(i,j),+_{i,j},0_{i,j})$ is defined for each $i,j \in \bbN$.
  If for each $i,j,k \in \bbN$, a map $\circ_{i,j,k}: \cR(j,k)\times \cR(i,j) \rightarrow \cR(i,k)$ is defined and satisfies the following conditions, $\cR$ is called a \emph{ringoid} (also called a \emph{small $\mathbf{Ab}$-enriched category}).
  \begin{enumerate}
    \item For each $i$, there exists an element $1_i \in \cR(i,i)$ such that $a \circ_{i,i,j} 1_i = a$, $1_i \circ_{j,i,i} b = b$ ($j \in \bbN$, $a \in \cR(i,j)$, $b \in \cR(j,i)$),
    \item $(a \circ_{j,k,l} b) \circ_{i,j,l} c = a \circ_{i,k,l} (b \circ_{i,j,k} c)$ ($a \in \cR(k,l)$, $b\in \cR(j,k)$, $c \in \cR(i,j)$),
    \item $(a+_{j,k}b)\circ_{i,j,k}c = a\circ_{i,j,k} c +_{i,k} b\circ_{i,j,k} c$ ($a,b\in \cR(j,k)$, $c\in\cR(i,j)$),
    \item $a\circ_{i,j,k}(b+_{i,j}c) = a\circ_{i,j,k}b +_{i,k} a\circ_{i,j,k} c$ ($a\in \cR(j,k)$, $b,c\in\cR(i,j)$),
    \item $a\circ_{i,j,k}0_{i,j} = 0_{i,k} = 0_{j,k}\circ_{i,j,k}b$ ($a\in\cR(j,k)$, $b\in\cR(i,j)$).
  \end{enumerate}
  We will omit subscripts of $+,\circ,0,1$ if there is no confusion.
\end{defi}
The notion of modules over a ring is extended to modules over a ringoid.
\begin{defi}
Let $\cR$ be a ringoid.
Suppose that for each $i\in \bbN$, an abelian group $(M(i),+_i,0_i)$ is defined.
If there is a map $\cdot_{i,j} : \cR(i,j)\times M(i) \rightarrow M(j)$ satisfying the following conditions, $M$ is called a \emph{left $\cR$-module}.
\begin{enumerate}
  \item $(a\circ_{i,j,k} b)\cdot_{i,k} x = a\cdot_{j,k}(b\cdot_{i,j}x)$ ($a\in\cR(j,k)$, $b\in\cR(i,j)$, $x\in M(i)$),
  \item $1_i\cdot_{i,i}x = x$ ($x\in M(i))$),
  \item $(a+_{i,j}b)\cdot_{i,j}x = (a\cdot_{i,j}x)+_{j}(b\cdot_{i,j}x)$ ($a,b\in\cR(i,j)$, $x\in M(i)$),
  \item $a\cdot_{i,j} (x +_i y) = (a\cdot_{i,j} x) +_j (a \cdot_{i,j} y)$ ($a \in \cR(i,j)$, $x,y \in M(i)$),
  \item $0_{i,j} \cdot_{i,j} x = 0_j$ ($x\in M(i)$).
\end{enumerate}

A \emph{right $\cR$-module} $M$ is also defined with a map $\cdot_{i,j} : M(i)\times \cR(i,j)\rightarrow M(j)$ in the same manner with right modules over a ring.

An \emph{$\cR$-linear map} $f : M \rightarrow M'$ between left $\cR$-modules $M, M'$ is a collection of group homomorphisms $f_i : M(i) \rightarrow M'(i)$ ($i\in \bbN$) that satisfy
\[
f_j(a\cdot_{i,j} x)=a\cdot_{i,j} f_i(x) \quad (a \in \cR(i,j),x \in M(i)).
\]
\end{defi}
Ringoids and modules over ringoids are originally defined in a category theoretic way (\cf~\cite{m72,mm16}).
(A ringoid is a small Ab-enriched category, and a module over a ringoid is an additive functor.)
Our definitions here are obtained by unfolding the category theoretic terminology in the original definitions so that those who are not familiar with category theory can understand them more easily.
\begin{defi}
  Let $\cR$ be a ringoid and $P$ be a family of sets $P_i$ ($i \in \bbN$).
  The free left $\cR$-module generated by $P$, denoted by $\freemod{\cR}{P}$ is defined as follows.
  For each $i \in \bbN$, $(\freemod{\cR}{P})(i)$ is the abelian group of formal finite sums
  \[
    \sum_{p_j\in P_j,\ j \in \bbN}a_{p_j}\underline{p_j},\quad (a_{p_j}\in \cR(j,i))
  \]
  and for each $r \in \cR(i,k)$,
  \[
    r\cdot\left(\sum_{p_j\in P_j,\ j\in \bbN}a_{p_j}\underline{p_j}\right)
    = \sum_{p_j\in P_j,\ j \in \bbN}(r\circ a_{p_j})\underline{p_j}.
  \]

If a left $\cR$-module $M$ is isomorphic to $\freemod{\cR}{P}$ for some $P$,
we say that $M$ is free.
\end{defi}
For $\imonring{\bbK}$, we write $C\underline x t$ for elements of $(\freemod{(\imonring{\bbK})}{P})(X)$ instead of $(C,t)\underline x$,
and $(D+C)\underline x t$ for $D\underline x t + C\underline x t$.

The tensor product of two modules over a ringoid is also defined.
\begin{defiC}[\cite{mm16}]
  Let $\cR$ be a ringoid, $M_1$ be a right $\cR$-module, and $M_2$ be a left $\cR$-module.
  For a family of groups $\{G_X\mid X \in P\}$ for some indexing set $P$, its direct sum, denoted by $\bigoplus_{X\in P}G_X$, is the subset of the direct product defined by $\{(g_X)_{X\in P}\in\prod_{X\in P}G_X\mid \text{$g_X=0$ except for finite $X$s}\}$.
  The direct sum of groups also forms a group.

  The tensor product $\tensor{M_1}{\cR}{M_2}$ is the quotient abelian group of
  $\bigoplus_{X\in\cR}\tensor{M_1(X)}{\cR(X,X)}{M_2(X)}$
  by relations $(x\cdot a)\otimes y - x\otimes (a\cdot y)$ for all $a \in \cR(Y,X)$, $x\in M(X), y\in M(Y)$.
\end{defiC}

We define an equivalence between two TRSs $(\Sigma,R)$, $(\Sigma',R')$, called \emph{Tietze equivalence}.
\begin{defi}
\label{def:tietze}
  Two TRSs are \emph{Tietze equivalent} if one is obtained from the other by applying a series of \emph{Tietze transformations} defined as follows:
\begin{enumerate}
  \item If $f^{(n)}$ is a symbol not in $\Sigma$ and $t\in T(\Sigma)$ has variables in $\{x_1,\dots,x_n\}$, then $(\Sigma,R)$ can be transformed into $(\Sigma\cup \{f\}, R\cup\{t \rrule f(x_1,\dots,x_n)\})$.
  \item If $t \rrule f(x_1,\dots,x_n)\in R$, $t \in T(\Sigma\setminus \{f\})$, and $f$ does not occur in any rule in $R'=R\setminus \{t\rrule f(x_1,\dots,x_n)\}$, then $(\Sigma,R)$ can be transformed into $(\Sigma\setminus\{f\}, R')$.
  \item If $t \xleftrightarrow{*}_R s$, then $(\Sigma,R)$ can be transformed into $(\Sigma,R\cup\{t\rrule s\})$.
  \item If $t\rrule s\in R$ and $t \xleftrightarrow{*}_{R'} s$ for $R'=R\setminus\{t\rrule s\}$, then $(\Sigma,R)$ can be transformed into $(\Sigma,R')$.
\end{enumerate}
\end{defi}
We can see that any two TRSs $(\Sigma,R_1)$,$(\Sigma,R_2)$ are Tietze equivalent if they are equivalent in the usual sense, $\rewrst{R_1}\,=\,\rewrst{R_2}$.
Tietze equivalence is originally introduced in group theory \cite[\S 11]{t1908}
and is also defined for monoids \cite[7.2]{bo93}.
\begin{exa}
  Consider the signature $\Sigma = \{+^{(2)},S^{(1)},0^{(0)}\}$ and the set $R$ of four rules
  \[
    0 + x \rrule x,\ x + 0 \rrule x,\ S(x) + y\rrule S(x + y),\ (x + y) + z \rrule x + (y + z).
  \]
  We can see $(\Sigma,R)$ is Tietze equivalent to $(\Sigma',R')$ where
  \[
   \Sigma' = \{+^{(2)},0^{(0)},1^{(0)}\},\ R' = \{0 + x \rrule x,\ x + 0 \rrule x,\ (x + y) + z \rrule x + (y + z)\}
   \]
   as follows:
   \begin{align*}
     (\Sigma,R)
     &\xlongrightarrow{\rm (1)} (\Sigma\uplus \{1^{(0)}\}, R\uplus\{S(0) \rrule 1\})\\
     &\xlongrightarrow{\rm (3)} (\Sigma\uplus\{1^{(0)}\}, R\uplus\{S(0) \rrule 1,\ 1 + x\rrule S(x)\})\\
     &\xlongrightarrow{\rm (4)} (\Sigma\uplus\{1^{(0)}\}, R\uplus\{1+x\rrule S(x)\})\\
     &\xlongrightarrow{\rm (4)} (\Sigma\uplus\{1^{(0)}\}, R\uplus\{1+x\rrule S(x)\}\setminus\{S(x)+y\rrule S(x+y)\})\\
     &\xlongrightarrow{\rm (2)} (\Sigma\uplus\{1^{(0)}\}\setminus\{S^{(1)}\}, R\setminus\{S(x)+y\rrule S(x+y)\})
     = (\Sigma',R').
   \end{align*}
\end{exa}

Now, we outline Malbos-Mimram's construction of the homology groups of TRSs.
Let $d=\deg(R)$.
\begin{enumerate}
  \item We begin by defining a new ringoid from $\imonring{\bbK}$.
  That ringoid, denoted by $\overline{\imonring{\bbK}}^{(\Sigma,R)}$,
  depends only on the Tietze equivalence class of $(\Sigma,R)$.
  $\overline{\imonring{\bbK}}^{(\Sigma,R)}$ corresponds to $\imonring{\facmon{\Sigma}{R}}$ in the case $(\Sigma,R)$ is an SRS.
  \item From this step, we write $\cR$ for $\overline{\imonring{\bbK}}^{(\Sigma,R)}$.
  It can be shown that we have a partial free resolution
  \begin{equation}\label{eqn:main_resol}
  \freemod{\cR}{\bfP_3}
  \xrightarrow{\partial_2} \freemod{\cR}{\bfP_2}
  \xrightarrow{\partial_1} \freemod{\cR}{\bfP_1}
  \xrightarrow{\partial_0} \freemod{\cR}{\bfP_0} \xrightarrow{\epsilon} \cZ
  \end{equation}
  where every $\bfP_i$ is a family of sets $(\bfP_i)_j$ given by $(\bfP_0)_1 = \{1\}$, $(\bfP_0)_j = \emptyset$ ($j \neq 1$), $(\bfP_1)_j = \Sigma^{(j)} = \{f\in\Sigma\mid\text{$f$ is of arity $j$}\}$, $(\bfP_2)_j = \{l\rrule r\in R\mid \text{$l$ is of arity $j$}\}$,
  $(\bfP_3)_j = \{((\square,A,s),(C,B,t))\text{ : critical pair} \mid \text{one of $A,B$ is in $(\bfP_2)_j$, and the other is in } (\bfP_2)_k\\ \text{ for } k\le j\}$.
  \item By taking the tensor product $\bbZ/d\bbZ\otimes_\cR-$, we have the chain complex
    \newcommand{\squeezexrightarrow}[1]{%
      \xrightarrow{%
        \!\!%
        \raisebox{1mm}{\ensuremath{\scriptstyle{#1}}}%
        \!\!%
      }%
    }
  \begin{equation} \label{eqn:main_complex}
  \tensor{\bbZ/d\bbZ}{\cR}{\freemod{\cR}{\bfP_3}} \squeezexrightarrow{\bbZ/d\bbZ\otimes\partial_2}
  \tensor{\bbZ/d\bbZ}{\cR}{\freemod{\cR}{\bfP_2}} \squeezexrightarrow{\bbZ/d\bbZ\otimes\partial_1}
  \tensor{\bbZ/d\bbZ}{\cR}{\freemod{\cR}{\bfP_1}} \squeezexrightarrow{\bbZ/d\bbZ\otimes\partial_0}
  \tensor{\bbZ/d\bbZ}{\cR}{\freemod{\cR}{\bfP_0}}
\end{equation}
  where $\bbZ/d\bbZ$ above is the $\cR$-module defined by $\bbZ/d\bbZ(i) = \bbZ/d\bbZ$ (the abelian group of integers) for each object $i$, and the scalar multiplication is given by $(C,t)\cdot k = k$.
  \item The homology groups can be defined by
  \[
  H_i(\Sigma,R) = \ker(\bbZ/d\bbZ\otimes\partial_{i-1})/\im(\bbZ/d\bbZ\otimes\partial_i).
  \]
  It is shown that the homology groups of TRS depend only on the Tietze equivalence class of $(\Sigma,R)$.
  Thus, we have the following:
  \[
  \rewrst{R_1}\,=\,\rewrst{R_2}\implies H_i(\Sigma,R_1) \cong H_i(\Sigma,R_2).
  \]
\end{enumerate}
For the step 1, we define the relations of $\overline{\imonring{\bbK}}^{(\Sigma,R)}$.
We identify elements in $\imonring{\bbK}$ as follows.
(a) For two $m$-tuples $t=\langle t_1,\dotsc,t_m\rangle, s=\langle s_1,\dotsc,s_m\rangle$ of terms, we identify $t$ and $s$ if $t \xleftrightarrow{*}_R s$.
(b) Similarly, for two single-hole contexts $C,D$, we identify $C$ and $D$ if $C\xleftrightarrow{*}_R D$.
For the last identification, we introduce operator $\kappa_i$ which takes a term $t$ and returns the formal sum of single-hole contexts $C_1+\dotsb+C_m$ where $C_j$ $(j=1,\dotsc,m)$ is obtained by replacing the $j$-th occurrence of $x_i$ with $\square$ in $t$, and $m$ is the number of the occurrences of $x_i$ in $t$.
For example, we have
\begin{align*}
  \kappa_1(f(g(x_1,x_2),x_1)) &= f(g(\square,x_2),x_1) + f(g(x_1,x_2),\square),\\
  \kappa_2(f(g(x_1,x_2),x_1)) &= f(g(x_1,\square),x_1),\\
  \kappa_2(h(x_1)) &= 0.
\end{align*}
The definition of $\kappa_i$ can be stated inductively as follows:
\begin{align*}
  \kappa_i(x_i) &= \square,\ \kappa_i(x_j) = 0 \quad (j\neq i),\\
  \kappa_i(f(t_1,\dotsc,t_n)) &= \sum_{k=1}^n f(t_1,\dotsc,t_{k-1},\kappa_i(t_k),t_{k+1},\dotsc,t_n).
\end{align*}
Then, (c) we identify formal sums of bicontexts $(C_1,t)+\dots+(C_k,t)$ and $(D_1,t)+\dots+(D_l,t)$ if $\kappa_i(u) = C_1 + \dots + C_k$, $\kappa_i(v) = D_1 + \dots + D_l$ for some positive integer $i$ and terms $u,v$ such that $u \xleftrightarrow{*}_R v$.
$\overline{\imonring{\bbK}}^{(\Sigma,R)}$ is defined as the quotient of $\imonring{\bbK}$ by the equivalence class generated by the identifications (a), (b), and (c).

We omit the definitions of the $\cR$-linear maps $\epsilon,\partial_i$ ($i = 0,1,2$) in the step 2, but we describe the group homomorphisms $\bbZ/d\bbZ\otimes\partial_i : \tensor{\bbZ/d\bbZ}{\cR}{\freemod{\cR}{\bfP_{i+1}}}\rightarrow \tensor{\bbZ/d\bbZ}{\cR}{\freemod{\cR}{\bfP_i}}$.
Let $\tilde\partial_i$ denote $\bbZ/d\bbZ\otimes\partial_i$ for simplicity.
For the step 2, we define the $\cR$-linear maps $\epsilon,\partial_i$ ($i = 0,1,2$).
For $f^{(n)}\in\Sigma$, the homomorphism $\tilde\partial_0 : \tensor{\bbZ/d\bbZ}{\cR}{\freemod{\cR}{\bfP_{1}}}\rightarrow \tensor{\bbZ/d\bbZ}{\cR}{\freemod{\cR}{\bfP_0}}$
is given by
\[
\tilde\partial_0(\underline f) = (n-1)\underline{1}.
\]
For a term $t$, we define $\varphi(t)$ as the linear combinaton of symbols $\sum_{f\in\Sigma}n_f\underline{f}$ where $n_f$ is the number of occurrences of $f$ in $t$.
Using this, for $l\rrule r\in R$, the homomorphism $\tilde\partial_1 :
\tensor{\bbZ/d\bbZ}{\cR}{\freemod{\cR}{\bfP_{2}}}\rightarrow \tensor{\bbZ/d\bbZ}{\cR}{\freemod{\cR}{\bfP_1}}$
is given by
\[
\tilde\partial_1(\underline{l\rrule r}) = \varphi(r) - \varphi(l).
\]
For a critical pair $((\square,l\rightarrow r,s),(C,u\rightarrow v,t))$, let $(D_i,l_i\rightarrow r_i,s_i)$, $(C_j,u_j\rightarrow v_j,t_j)$ ($i=1,\dots,k,j=1,\dots,l$) be rewriting steps such that $r\circ s = D_1[l_1\circ s_1], D_1[r_1\circ s_1]=D_2[l_2\circ s_2],\dots,D_{k-1}[r_{k-1}\circ s_{k-1}] = D_k[l_k\circ s_k]$, $C[v\circ t] = C_1[u_1\circ t_1],C_1[v_1\circ t_1]=C_2[u_2\circ t_2],\dots,C_{l-1}[v_{l-1}\circ t_{l-1}]=C_{l}[u_l\circ t_l]$, $D_k[r_k\circ s_k] = C_l[v_l\circ t_l]$.
Then the map $\tilde\partial_2((\square,l\rightarrow r,s),(C,u\rightarrow v,t))$ is defined by
\[
\underline{u\rrule v} - \underline{l\rrule v} - \sum_{i=1}^k\underline{u_i\rrule v_i} - \sum_{j=1}^l\underline{l_j\rrule r_j}.
\]

Malbos-Mimram's lower bound for the number of rewrite rules is given by $s(H_2(\Sigma,R))$.
(Recall that $s(G)$ denotes the minimum number of generators of an abelian group $G$.)
More precisely, $\#\Sigma' \ge s(H_1(\Sigma,R))$ and $\#R' \ge s(H_2(\Sigma,R))$ hold for any TRS $(\Sigma',R')$ that is Tietze equivalent to $(\Sigma,R)$.
These inequalities are shown in a similar way to the proof of Theorem \ref{thm:srs_ineq}.

%% file: mainthm.tex
\label{section:mainthm}
Let $(\Sigma,R)$ be a complete TRS.
We first simplify the tensor product $\tensor{\bbZ/d\bbZ}{\cR}{\freemod{\cR}{\bfP_i}}$ in (\ref{eqn:main_complex}).
\begin{lem}
  Let $d = \deg(R)$ and $P$ be a family of sets $P_0,P_1,\dots$.
  Then, we have
  $\tensor{\bbZ/d\bbZ}{\cR}{\freemod{\cR}{P}} \cong \freemod{(\bbZ/d\bbZ)}{\biguplus_iP_i}$.
  Especially, if $d=0$, $\tensor{\bbZ}{\cR}{\freemod{\cR}{P}} \cong \freemod{\bbZ}{\biguplus_iP_i}$.
\end{lem}
\begin{proof}
  We define a group homomorphism $f :\tensor{\bbZ/d\bbZ}{\cR}{\freemod{\cR}{P}} \rightarrow \freemod{(\bbZ/d\bbZ)}{\biguplus_iP_i}$ by
  $f((w_n)_{n\ge 0}) = \sum_{n\ge0}f_n(w_n)$ where $f_n : \tensor{\bbZ/d\bbZ}{\cR(n,n)}{\freemod{\cR}{P}(n)} \rightarrow \freemod{(\bbZ/d\bbZ)}{P_n}$ is defined by
  $f_n([k]\otimes C\underline{a}t) = [k]\underline{a}$ for $a\in P_n$.
  Since each $f_n$ is an isomorphism, $f$ is also an isomorphism.
\end{proof}
As special cases of this lemma, we have $\tensor{\bbZ/d\bbZ}{\cR}{\freemod{\cR}{\bfP_0}}\cong \freemod{(\bbZ/d\bbZ)}{\Sigma}$,
$\tensor{\bbZ/d\bbZ}{\cR}{\freemod{\cR}{\bfP_1}} \cong \freemod{(\bbZ/d\bbZ)}{R}$,
and
$\tensor{\bbZ/d\bbZ}{\cR}{\freemod{\cR}{\bfP_2}} \cong \freemod{(\bbZ/d\bbZ)}{\CR(R)}$.
Additionally, we can see each group homomorphism $\tilde\partial_i$ ($i=0,1,2$) is a $\bbZ/d\bbZ$-linear map.

To prove Theorem \ref{thm:main_explicit}, we show the following lemma.
\begin{lem} \label{lemma:smith_hom}
  Let $d = \deg(R)$.
  If $d=0$ or $d$ is prime,
  $\#R - e(R) = s(H_2(\Sigma,R)) + s(\im\tilde\partial_1)$.
(Recall that $s(G)$ is the minimum number of generators of a group $G$.)
\end{lem}
\begin{proof}
  By definition, $D(R)$ defined in Section \ref{section:algorithm} is a matrix representation of $\tilde\partial_2$.
  Suppose $d$ is prime.
  In this case, $s(H_2(\Sigma,R))$ is equal to the dimension of $H_2(\Sigma,R)$ as a $\bbZ/d\bbZ$-vector space.
  By the rank-nullity theorem, we have
  \begin{align*}
    \dim(H_2(\Sigma,R))&=\dim(\ker\tilde\partial_1)-\dim(\im\tilde\partial_2)\\
    &=\dim(\tensor{\bbZ/d\bbZ}{\cR}{\freemod{\cR}{\bfP_1}})-\dim(\im\tilde\partial_1)-\dim(\im\tilde\partial_2)\\
    &=\dim(\freemod{(\bbZ/d\bbZ)}{R}) - \dim(\im\tilde\partial_1) - \rank(D(R))\\
    &=\#R - \dim(\im\tilde\partial_1) - e(R).
  \end{align*}

  Suppose $d=0$.
  We show $H_2(\Sigma,R)\cong \bbZ^{\#R-r-k}\times\bbZ/e_1\bbZ\times\dotsb\times\bbZ/e_r\bbZ$ where $r = \rank(D(R))$, $k=s(\im\tilde\partial_1)$, and $e_1,\dotsc,e_r$ are the elementary divisors of $D(R)$.
  Let
  \[
  \overline\partial_1 : \tensor{\bbZ}{\cR}{\freemod{\cR}{\bfP_1}}/\im\tilde\partial_2 \rightarrow \tensor{\bbZ}{\cR}{\freemod{\cR}{\bfP_0}}
  \]
  be the group homomorphism defined by $[x]\mapsto \tilde\partial_1(x)$.
  $\overline\partial_1$ is well-defined since $\im\tilde\partial_2\subset\ker\tilde\partial_1$, and
  $\ker\overline\partial_1$ is isomorphic to  $\ker\tilde\partial_1/\im\tilde\partial_2=H_2(\Sigma,R)$.
  By taking the basis $v_1,\dotsc,v_{\#R}$ of $\tensor{\bbZ}{\cR}{\freemod{\cR}{\bfP_1}}\cong \freemod{\bbZ}{R}$ such that $D(R)$ is the matrix representation of $\tilde\partial_2$ under the basis $v_1,\dotsc,v_{\#R}$ and some basis of $\tensor{\bbZ}{\cR}{\freemod{\cR}{\bfP_2}}$,
  we can see $\tensor{\bbZ}{\cR}{\freemod{\cR}{\bfP_1}}/\im\tilde\partial_2\cong \bbZ^{\#R-r}\times \bbZ/e_1\bbZ\times\dotsb\times\bbZ/e_k\bbZ$.
  Suppose $\overline\partial_1(e_i[x]) = 0$ for some $x$ and $i = 1,\dotsc,r$.
  Since $\overline\partial_1$ is a homomorphism, $\overline\partial_1(e_i[x]) = e_i\overline\partial_1([x]) \in \tensor{\bbZ}{\cR}{\freemod{\cR}{\bfP_0}} \cong \freemod{\bbZ}{\Sigma}$ holds.
  Since $\freemod{\bbZ}{\Sigma}$ is free, we have $[x] = 0$.
  Therefore, $\ker\overline\partial_1$ is included in the subset of $\tensor{\bbZ}{\cR}{\freemod{\cR}{\bfP_1}}/\im\tilde\partial_2$ isomorphic to $\bbZ^{\#R-r}\times \{0\}\times\dotsb\times\{0\}$.
  Thus, $\ker\overline\partial_1 \cong \bbZ^{\#R-r-k}\times\bbZ/e_1\bbZ\times\dotsb\times\bbZ/e_r\bbZ$.

  Since $\bbZ/e\bbZ \cong 0$ if $e$ is invertible, $\bbZ^{\#R-r-k}\times\bbZ/e_1\bbZ\times\dotsb\times\bbZ/e_k\bbZ \cong \bbZ^{\#R-r-k}\times\bbZ/e_{e(R)+1}\bbZ\times\dotsb\bbZ/e_r\bbZ =: G$.
  The group $G$ is generated by $(\underbrace{1,0,\dotsc,0}_{\#R-r-k},\underbrace{[0],\dotsc,[0]}_{r-e(R)})$, $\dotsc$, $(0,\dotsc,0,1,[0],\dotsc,[0])$, $\dotsc$, $(0,\dotsc,0,[1],[0],\dotsc,[0])$, $\dotsc$, $(0,\dotsc,0,[0],\dotsc,[0],[1])$,
  so we have $s(G) \le \#R - r - k + r - e(R) = \#R - k - e(R)$.
  Let $p$ be a prime number which divides $e_{e(R)+1}$.
  We can see $G/pG \cong (\bbZ/p\bbZ)^{\#R-k-e(R)}$.
  It is not hard to see $s(G)\ge s(G/pG)$, and
  since $G/pG$ is a $\bbZ/p\bbZ$-vector space, $s(G/pG)=\dim(G/pG) = \#R-k-e(R)$.
  Thus, $s(H_2(\Sigma,R)) = s(G) =  \#R - s(\im\tilde\partial_1) - e(R)$.
\end{proof}
By Lemma \ref{lemma:smith_hom}, Theorem \ref{thm:main_explicit} is implied by the following theorem:
\begin{thm}\label{thm:main}
  Let $(\Sigma,R)$ be a TRS and $d = \deg(R)$. If $d=0$ or $d$ is prime,
    \begin{equation}
    \label{eqn:main}
    \#R \ge  s(H_2(\Sigma,R))+ s(\im \tilde\partial_1).
  \end{equation}
\end{thm}
\begin{proof}
  By the first isomorphism theorem, we have an isomorphism between $\bbZ/d\bbZ$-modules
  \[
  \im\tilde\partial_1 \simeq \tensor{\bbZ/d\bbZ}{\cR}{\freemod{\cR}{\bbP_2}}\quot \ker\tilde\partial_1
  \]
  and by the third isomorphism theorem, the right hand side is isomorphic to
  \begin{align*}
  \tensor{\bbZ/d\bbZ}{\cR}{\freemod{\cR}{\bbP_2}}\quot \ker\tilde\partial_1
  &\simeq \left(\tensor{\bbZ/d\bbZ}{\cR}{\freemod{\cR}{\bbP_2}}\quot\im\tilde\partial_2\right)
  \quot \left(\ker\tilde\partial_1\quot\im\tilde\partial_2\right)\\
  &\simeq
  \left(\tensor{\bbZ/d\bbZ}{\cR}{\freemod{\cR}{\bbP_2}}\quot\im\tilde\partial_2\right)
  \quot H_2(\Sigma,R).
\end{align*}
  Thus, we obtain the following exact sequence by Proposition \ref{thm:exactbasic}:
  \[
  0 \rightarrow H_2(\Sigma,R)
   \rightarrow
   \tensor{\bbZ/d\bbZ}{\cR}{\freemod{\cR}{\bbP_2}}\quot\im\tilde\partial_2
   \rightarrow
   \im\tilde\partial_1
   \rightarrow 0.
  \]
  By Theorem \ref{thm:subfree}, since $\im\tilde\partial_1\subset \tensor{\bbZ/d\bbZ}{\cR}{\freemod{\cR}{\bbP_1}}\cong \freemod{(\bbZ/d\bbZ)}{R}$ and $\freemod{(\bbZ/d\bbZ)}{R}$ is a free $\bbZ/d\bbZ$-module, $\im\tilde\partial_1$ is also free and by Proposition \ref{thm:split},
  we have $\tensor{\bbZ/d\bbZ}{\cR}{\freemod{\cR}{\bbP_2}}\quot\im\tilde\partial_2 \cong H_2(\Sigma,R)\times \im\tilde\partial_1$.
  Therefore, $s(\tensor{\bbZ/d\bbZ}{\cR}{\freemod{\cR}{\bbP_2}}\quot\im\tilde\partial_2) = s(H_2(\Sigma,R)) + s(\im\tilde\partial_1)$.
  Since $\tensor{\bbZ/d\bbZ}{\cR}{\freemod{\cR}{\bbP_2}}\quot\im\tilde\partial_2$ is generated by $[l_1\rrule r_1],\dots,[l_k\rrule r_k]$ if $R = \{l_1\rrule r_1,\dots,l_k\rrule r_k\}$, we obtain
  \begin{equation*}
  k = \#R\ge s(\tensor{\bbZ/d\bbZ}{\cR}{\freemod{\cR}{\bbP_2}}\quot\im\tilde\partial_2) = s(H_2(\Sigma,R))+s(\im\tilde\partial_1).
  \end{equation*}
  Thus, we get (\ref{eqn:main}).
\end{proof}
Now, we prove our main theorem, Theorem \ref{thm:main_explicit}.
\begin{proof}[Proof of Theorem \ref{thm:main_explicit}]
As we stated, $H_2(\Sigma,R)$ depends only on the Tietze equivalence class of $(\Sigma,R)$ and particularly, $H_2(\Sigma,R')$ is isomorphic to $H_2(\Sigma,R)$ if $R'$ is equivalent to $R$ (in the sense ${\xleftrightarrow{*}_R}={\xleftrightarrow{*}_{R'}}$).
Let us show $s(\im\tilde\partial_1)$ depends only on the  equivalence class of $R$.
For a left $\fR$-module $M$, $\rank(M)$ denotes the cardinality of a minimal linearly independent generating set of $M$, that is, a minimal generating set $S$ of $G$ such that any element $s_1,\dotsc,s_k\in\Gamma$,
and $r_1s_1 + \dotsb + r_ks_k = 0 \implies r_1=\dotsb=r_k=0$ for any $r_1,\dotsc,r_k\in\fR$, $s_1,\dotsc,s_k \in S$.
It can be shown that $\rank(M) = s(M)$ if $M$ is free.
Especially, $s(\im\tilde\partial_1) = \rank(\im\tilde\partial_1)$ since $\im\tilde\partial_1\subset\freemod{\bbZ}{R}$ if $\deg(R)=0$.
Also, $\rank(\im\tilde\partial_1) = \rank(\ker\tilde\partial_0) - \rank(\ker\tilde\partial_0/\im\tilde\partial_1)$ is obtained by a general theorem \cite[Ch 10, Lemma 10.1]{r10}.
By definition, $\tilde\partial_0$ does not depend on $R$.
Since $\ker\tilde\partial_0/\im\tilde\partial_1 = H_1(\Sigma,R)$ depends only on the Tietze quivalence class of $R$, two sets of rules $R,R'$ with ${\xleftrightarrow{*}_R}={\xleftrightarrow{*}_{R'}}$ give the same $\rank(\im\tilde\partial_1)$.

In conclusion, for any TRS $R'$ equivalent to $R$, we obtain
$\#R' \ge s(H_2(\Sigma,R))+ s(\im \tilde\partial_1) = \#R - e(R)$.
\end{proof}

%% file: prime.tex
\label{section:prime}
Let $(\Sigma,R)$ be a complete TRS.
It is known that in confluence tests (and then in Knuth-Bendix completion), it suffices to consider only prime critical pairs \cite{k88}.
(A critical pair $r_1\sigma \leftarrow l_1\sigma= C[l_2\sigma]\rightarrow C[r_2\sigma])$ is prime if no proper subterm of $l_2\sigma$ is reducible by $R$.)
We have defined the matrix $D(R)$ using the critical pairs of $R$, but
in fact, we can restrict the critical pairs to the prime ones and obtain the same $e(R)$.
In other words, we have the following theorem.
\begin{thm}
In the matrix $D(R)$, all columns corresponding to critical pairs that are not prime can be transformed into the zero vectors by elementary column operations.
\end{thm}
\begin{proof}
For any terms $t,s$ and position $p \in \pos(t)$, we write $t[s]_p$ for the term obtained from $t$ by replacing the subterm at $p$ with $s$.
Suppose that $R=\{l_1\rrule r_1,\dots,l_n\rrule r_n\}$ and that there is a non-prime critical pair $(r_i\sigma,l_i\sigma[r_j\sigma]_p)$ for some position $p\in\pos(l_i)$.
Then, there is a rule $l_k \rrule r_k$ such that $l_k$ matches $(l_j\sigma)|_p$ for some position $p' \in \pos(l_j\sigma)$.
So, we have the following three paths.
\begin{equation}
\begin{tikzcd}
\label{diagram:critical_triple}
l_i\sigma \arrow[equal]{r} \arrow[d] & {l_i\sigma[l_j\sigma]_p} \arrow[equal]{r} \arrow[d] & {l_i\sigma[l_j\sigma[l_k\sigma]_{p'}]_p} \arrow[d]\\
r_i\sigma \arrow[d]                   & {l_i\sigma[r_j\sigma]_p} \arrow[d]                                & {l_i\sigma[l_j\sigma[r_k\sigma]_{p'}]_p} \arrow[d] \\
\vdots \arrow[rd]                     & \vdots \arrow[d]                                        & \vdots \arrow[ld]           \\
                                      & t                                                       &                             
\end{tikzcd}
\end{equation}
We show that the column for $(r_i\sigma,l_i\sigma[r_j\sigma]_p)$ in $D(R)$ can be transformed into the zero vector by elementary column operations.
Let $\pos_{\mathcal{F}}(s)$ be the set $\pos(s)\setminus\{\text{variable positions in $s$}\}$ for a term $s$.
Consider the following four cases: (1) $p' \in \pos_{\mathcal{F}}(l_j)$ and $pp' \in \pos_{\mathcal{F}}(l_i)$, (2) $p' \in \pos_{\mathcal{F}}(l_j)$ and $pp'\notin \pos_{\mathcal{F}}(l_i)$, (3) $p'\notin \pos_{\mathcal{F}}(l_j)$ and $pp'\in \pos_{\mathcal{F}}(l_i)$, (4) $p' \notin \pos_{\mathcal{F}}(l_j)$ and $pp'\notin \pos_{\mathcal{F}}(l_i)$.
\begin{enumerate}
    \item Case where $p' \in \pos_{\mathcal{F}}(l_j)$ and $pp'\in\pos_{\mathcal{F}}(l_i)$:
In this case, for some substitutions $\sigma',\sigma''$ more general than $\sigma$, the pairs $P_{j,k} = (r_j\sigma', l_j\sigma'[r_k\sigma']_{p'})$ and $P_{i,k}=(r_i\sigma'', l_i\sigma''[l_j\sigma''[r_k\sigma'']_{p'}]_p)$ are critical.
Let $P_{i,j}=(r_i\sigma,l_i\sigma[r_j\sigma]_p)$.
Also, we write $v_{i,m}, v_{j,m}, v_{k,m}$ for the numbers of $l_m\rrule r_m$ appears in $l_i\sigma\rrule r_i\sigma \rrule \dots \rrule t$, $l_i\sigma[l_j\sigma]_p \rrule l_i\sigma[r_j\sigma]_p \rrule \dots \rrule t$, $l_i\sigma[l_j\sigma[l_k\sigma]_{p'}]_p \rrule l_i\sigma[l_j\sigma[r_k\sigma]_{p'}]_p \rrule \dots \rrule t$, respectively.
Then, the column of the matrix $D(R)$ for the critical pair $P_{\alpha,\beta}$ ($\alpha,\beta\in\{i,j,k\}$) can be given by
$V_{\alpha,\beta} := (v_{\alpha,1}-v_{\beta,1}~v_{\alpha,2}-v_{\beta,2}~\dots~v_{\alpha,n}-v_{\beta,n})^T$.
So, we have $V_{i,j} = V_{i,k} - V_{j,k}$ and this means that the column for $P_{i,j}$ can be transformed into the zero vector by an elementary column operation.
    \item Case where $p' \in \pos_{\mathcal{F}}(l_j)$ and $pp'\notin\pos_{\mathcal{F}}(l_i)$:
In this case, $(r_j\sigma',l_j\sigma'[r_k\sigma']_{p'})$ is a critical pair for some $\sigma'$ and $l_k\rrule r_k$ can rewrite $x_a\sigma$ for some variable $x_a$ in $l_i$ into some term $s$.
Then, we have the two paths:
\[
  \begin{tikzcd}
{l_i\sigma[l_j\sigma[l_k\sigma]_{p'}]_p} \arrow[d, Rightarrow, no head] \arrow[r, "l_k \rrule r_k"]
\ar[start anchor=center, end anchor=center, yshift=1.2em, decorate,decoration={brace,amplitude=2mm}, no head]{rrr}[outer sep=2mm]{\#_a l_i \text{ times}}
& {l_i\sigma[l_j\sigma[r_k\sigma]_{p'}]_p} \arrow[r, "l_k \rrule r_k"] & \dots \arrow[r, "l_k \rrule r_k"] & l_i\sigma'' \arrow[d, "l_i \rrule r_i"] \\
l_i\sigma \arrow[r, "l_i \rrule r_i"]
& r_i\sigma \arrow[r, "l_k \rrule r_k"]
\ar[start anchor=center, end anchor=center, yshift=-1.2em, decorate,decoration={brace, mirror,amplitude=2mm}, no head]{rr}[outer sep=2mm,below]{\#_ar_i \text{ times}}
& \dots \arrow[r, "l_k \rrule r_k"] & r_i\sigma''          
\end{tikzcd}
\]
where $\sigma''$ is the substitution that is the same as $\sigma$ but $x_a\sigma'' = s$.
Then, $v_{i,m}-v_{k,m} = 0$ for any $m\neq k$ and $v_{i,k}-v_{k,k}=\#_al_i-\#_ar_i = 0 \mod \deg(R)$.
Therefore, we have $V_{i,j} = -V_{j,k}\mod \deg(R)$.
    \item Case where $p'\notin\pos_{\mathcal{F}}(l_j)$ and $pp'\in \pos_{\mathcal{F}}(l_i)$:
We can show $V_{i,j}=V_{i,k}$ in a way similar to (2).
    \item Case where $p'\notin\pos_{\mathcal{F}}(l_j)$ and $pp'\notin \pos_{\mathcal{F}}(l_i)$: Also in a similar way, we can see that $V_i,j$ is the zero vector.
\end{enumerate}
Thus, we can remove the column for $P_{i,j}$ in $D(R)$.
Repeating this proccess for all non-prime pairs, we obtain the desired result.
\end{proof}

For case (1) in the proof, we can actually remove the column for $P_{i,k}$ or $P_{j,k}$ instead of $P_{i,j}$ by symmetry.
We shall call a triple of critical pairs like (\ref{diagram:critical_triple}) a \emph{critical triple}.
Recall that we mentioned $D(R)$ is a matrix presentation of $\tilde\partial_2=\bbZ/d\bbZ\otimes \partial_2 : \tensor{\bbZ/d\bbZ}{\cR}{\freemod{\cR}{\bfP_3}} \rightarrow \tensor{\bbZ/d\bbZ}{\cR}{\freemod{\cR}{\bfP_2}}$ in Section \ref{section:mainthm}.
If we define $\bfP_4$ to be a collection of critical triples and
$\tilde\partial_3 : \tensor{\bbZ/d\bbZ}{\cR}{\freemod{\cR}{\bfP_4}}\rightarrow\tensor{\bbZ/d\bbZ}{\cR}{\freemod{\cR}{\bfP_3}}$
to be $\tilde\partial_3(\underline{(P_{i,j},P_{i,k},P_{j,k})}) = \underline{P_{i,j}}-\underline{P_{i,k}}+\underline{P_{j,k}}$,
then we have $\tilde\partial_2\circ\tilde\partial_3(\underline{(P_{i,j},P_{i,k},P_{j,k})}) = \tilde\partial_2(\underline{P_{i,j}}-\underline{P_{i,k}}+\underline{P_{j,k}})=0$.
(This corresponds to $V_{i,k}=V_{i,k}-V_{j,k}$.)
Therefore $\ker\tilde\partial_2 \supset\im\tilde\partial_3$ holds, so we can extend our chain complex (\ref{eqn:main_complex}):
\begin{equation*}
\tensor{\bbZ/d\bbZ}{\cR}{\freemod{\cR}{\bfP_4}} \xrightarrow{\tilde\partial_3}
\tensor{\bbZ/d\bbZ}{\cR}{\freemod{\cR}{\bfP_3}} \xrightarrow{\tilde\partial_2}
\tensor{\bbZ/d\bbZ}{\cR}{\freemod{\cR}{\bfP_2}} \xrightarrow{\tilde\partial_1}
\tensor{\bbZ/d\bbZ}{\cR}{\freemod{\cR}{\bfP_1}} \xrightarrow{\tilde\partial_0}
\tensor{\bbZ/d\bbZ}{\cR}{\freemod{\cR}{\bfP_0}}.
\end{equation*}
Note that since we have not defined $\partial_3 : \tensor{\bbZ/d\bbZ}{\cR}{\freemod{\cR}{\bfP_3}} \to \tensor{\bbZ/d\bbZ}{\cR}{\freemod{\cR}{\bfP_2}}$, the third homology $H_3$ is meaningless unless we define $\partial_3$ extending resolution (\ref{eqn:main_resol}) and show $\tilde\partial_3 = \bbZ_d\otimes\partial_3$.
However, this suggests that the next term of our partial resolution could be generated by critical triples.

%% file: impossible.tex
\label{section:impossible}
We consider the case where every symbol in $\Sigma$ is of arity 1.
Notice that any TRS $(\Sigma,R)$ can be seen as an SRS and $\deg(R)=0$ in this case.
We have $\rank(\ker \tilde\partial_0) = \#\Sigma$ since $\tilde\partial_0(\underline{f}) = 0$ for any $f \in \Sigma$.
Therefore, (\ref{eqn:main}) can be rewritten to
\begin{equation}
\label{eqn:morse}
  \#R - \#\Sigma \ge s(H_2(\Sigma,R))-\rank(H_1(\Sigma,R)).
\end{equation}
So, for SRSs, we have a lower bound of the difference between the number of rewrite rules and the number of symbols.
For groups, in fact, this inequality is proved in terms of group homology \cite{e61} without using the notion of a rewriting system.
In group theory, a group presentation $(\Sigma,R)$ is called \emph{efficient} if the equality of (\ref{eqn:morse}) holds and a group is called efficient if it has an efficient presentation.
It is known that inefficient groups exist \cite{s65}.
Let us move back to the case of general TRSs.
We have already seen that there exists a TRS such that none of its equivalent TRS satisfies the equality of (\ref{eqn:main}) in the last paragraph of Section \ref{section:algorithm}.
The \emph{deficiency} of (the equivalence class of) a TRS $(\Sigma,R)$, denoted by $\defc\langle\Sigma,R\rangle$, is the minimum of $\#R'-\#\Sigma'$ over all TRSs $(\Sigma',R')$ Tietze equivalent to $(\Sigma,R)$.
We pose the problem to decide inequalities of the deficiency for TRSs and see its undecidability is shown by using powerful facts from group theory.
\begin{prob}
\label{prob:def_bound}
Given an integer and a TRS $(\Sigma,R)$, does $\defc\langle \Sigma,R\rangle \le n$ hold?
\end{prob}
We will prove that Problem \ref{prob:def_bound} is undecidable.
It suffices to restrict the problems to the case where $n$ is negative and $(\Sigma,R)$ presents a group of finite index, that is, $(\Sigma,R)$ is an SRS and $\facmon{\Sigma}{R}=\Sigma^*/{\xleftrightarrow{*}_R}$ forms a group of finite index.
\begin{prob}
\label{prob:def_bound_gp}
Given a negative integer $n$ and an SRS $(\Sigma,R)$ whose corresponding monoid $\facmon{\Sigma}{R}$ forms a group of finite index, does $\defc\langle \Sigma,R\rangle \le n$ hold?
\end{prob}
\begin{thm}
\label{thm:undec}
Problem \ref{prob:def_bound_gp} is undecidable,
and then so is Problem \ref{prob:def_bound}
\end{thm}
To prove the theorem, we will apply one of the most useful tools on computability in group theory called Adian-Rabin theorem
which states every ``Markov property'' is undecidable.
\begin{defi}
Let $P$ be a property of finitely presented groups which is preserved under group isomorphism.
The property $P$ is said to be a \emph{Markov property} if
\begin{enumerate}
\item there exists a finitely presented group $G_{+}$ with $P$, and
\item there exists a finitely presented group $G_{-}$ such that there is no injective group homomorphism from $G_{-}$ to a finitely presented group $G$ with $P$.
\end{enumerate}
\end{defi}
The condition (2) is equivalent to the case where there exists a finitely presented group $G_{-}$ which does not have $P$ and whenever a finitely generated group $G$ has $P$, all subgroups of $G$ also have $P$.
\begin{thmC}[{\cite{a55}\cite{r58}\cite[Theorem 4.1]{ls01}}]
Markov properties are undecidable.
\end{thmC}
\begin{proof}[Proof of Theorem \ref{thm:undec}]
Let $n$ be a negative integer.
If a group $G$ is presented by $(\Sigma,R)$, we write $\defc G$ for $\defc\langle \Sigma,R\rangle$.
We show that $\defc G\le n$ is a Markov property.
Since the free group $F_k$ satisfies $\defc F_k = 0 - k = -k$ for any $k\ge 0$,
there always exist $G_{+}$ with $\defc G_{+} \le n$ and $G_{-}$ with $G_{-} > n$.
Therefore it is enough to show that for any finitely presented group $G$ of finite index and a subgroup $H$ of $G$, $\defc G \le n$ implies $\defc H \le n$.
Let $G$ be a finitely presented group with $\defc G \le n$ and $H$ be a subgroup of $G$.
Given a finite presentation $(\Sigma,R)$ of $G$, it is known that we can construct a presentation $(\Sigma',R')$ of $H$ satisfying
$\#R'-\#\Sigma' + 1 = [G:H](\#R-\#\Sigma + 1)$
where $[G:H]$ is the index of $H$ in $G$.
(See \cite[Ch. II. Proposition 4.1]{ls01}, for example.)
The way of construction is known as \emph{Reidemeister-Schreier method}.
Thus, we have
\[
\defc H +1 \le [G:H](\defc G + 1) \le \defc G + 1 \le n + 1.
\tag*{\qedhere}
\]
\end{proof}

%% file: concl.tex
We have seen that the number of rewrite rules is bounded below by a computable number defined using homology groups of TRSs.
The computation is by simple term rewriting and matrix transformation.
The fact that the theory of groups must have at least two equational axioms is proved as a corollary.
We have also showed that deciding $\defc \langle\Sigma,R\rangle \le n$ is not computationally possible.

%% file: group.tex
For the TRS $R$ defined in Example \ref{example:group},
$D(R)$ is given by the transpose of
\begin{footnotesize}
\[
\left(
\begin{array}{cccccccccc}
1 & 0 & 0 & 0 & 0 & 0 & 0 & 0 & 0 & 0 \\
0 & 0 & 0 & 0 & 0 & 0 & 0 & 0 & 0 & 0 \\
0 & 0 & 0 & 0 & 0 & 0 & 0 & 0 & 0 & 0 \\
0 & 0 & 0 & 0 & 0 & 0 & 0 & 0 & 1 & 1 \\
0 & 0 & 0 & 0 & 0 & 0 & 0 & 0 & 0 & 0 \\
0 & 0 & 0 & 0 & 0 & 1 & 0 & 0 & 0 & 1 \\
1 & 1 & 0 & 0 & 1 & 1 & 0 & 0 & 0 & 0 \\
1 & 1 & 0 & 1 & 0 & 0 & 0 & 0 & 1 & 0 \\
1 & 0 & 0 & 0 & 0 & 0 & 0 & 0 & 1 & 1 \\
1 & 1 & 1 & 0 & 0 & 0 & 0 & 0 & 0 & 0 \\
1 & 0 & 0 & 0 & 0 & 0 & 0 & 0 & 0 & 0 \\
1 & 0 & 0 & 0 & 0 & 0 & 0 & 0 & 0 & 0 \\
0 & 1 & 1 & 0 & 0 & 0 & 1 & 0 & 0 & 1 \\
0 & 0 & 0 & 0 & 0 & 0 & 1 & 0 & 1 & 0 \\
0 & 0 & 0 & 0 & 0 & 1 & 1 & 0 & 0 & 0 \\
0 & 1 & 0 & 1 & 0 & 0 & 1 & 0 & 0 & 0 \\
0 & 1 & 1 & 0 & 0 & 0 & 0 & 0 & 0 & 0 \\
0 & 1 & 1 & 0 & 0 & 0 & 1 & 0 & 0 & 1 \\
0 & 0 & 1 & 1 & 0 & 0 & 0 & 0 & 1 & 0 \\
0 & 0 & 1 & 0 & 1 & 1 & 0 & 0 & 0 & 0 \\
0 & 0 & 1 & 0 & 1 & 0 & 1 & 0 & 0 & 0 \\
1 & 0 & 0 & 0 & 0 & 0 & 0 & 0 & 1 & 1 \\
0 & 0 & 0 & 1 & 1 & 0 & 1 & 0 & 0 & 1 \\
0 & 0 & 1 & 1 & 0 & 0 & 0 & 1 & 1 & 0 \\
0 & 0 & 0 & 1 & 1 & 0 & 0 & 1 & 0 & 0 \\
0 & 1 & 0 & 1 & 0 & 0 & 1 & 0 & 0 & 0 \\
0 & 0 & 1 & 1 & 0 & 1 & 0 & 0 & 0 & 0 \\
1 & 0 & 0 & 0 & 0 & 1 & 0 & 0 & 0 & 1 \\
0 & 0 & 0 & 1 & 1 & 0 & 1 & 0 & 0 & 1 \\
0 & 0 & 1 & 0 & 1 & 0 & 0 & 0 & 1 & 0 \\
0 & 0 & 0 & 1 & 1 & 0 & 0 & 1 & 0 & 0 \\
0 & 1 & 0 & 0 & 1 & 0 & 1 & 0 & 0 & 0 \\
0 & 0 & 1 & 0 & 1 & 1 & 0 & 1 & 0 & 0 \\
0 & 0 & 0 & 0 & 0 & 0 & 0 & 1 & 0 & 0 \\
0 & 0 & 0 & 0 & 0 & 1 & 0 & 0 & 0 & 1 \\
1 & 0 & 1 & 0 & 1 & 1 & 0 & 1 & 0 & 0 \\
0 & 0 & 0 & 0 & 0 & 1 & 0 & 0 & 1 & 0 \\
0 & 0 & 0 & 0 & 0 & 1 & 0 & 0 & 1 & 0 \\
0 & 0 & 0 & 0 & 0 & 1 & 0 & 1 & 1 & 0 \\
0 & 0 & 0 & 0 & 0 & 1 & 1 & 0 & 0 & 0 \\
0 & 0 & 0 & 0 & 0 & 0 & 1 & 0 & 1 & 0 \\
0 & 0 & 0 & 0 & 0 & 0 & 0 & 1 & 0 & 0 \\
0 & 0 & 0 & 0 & 0 & 0 & 0 & 0 & 0 & 0 \\
0 & 0 & 0 & 0 & 0 & 0 & 0 & 1 & 0 & 0 \\
0 & 0 & 0 & 0 & 0 & 1 & 0 & 1 & 1 & 0 \\
0 & 0 & 0 & 0 & 0 & 0 & 0 & 1 & 0 & 0 \\
0 & 0 & 0 & 0 & 0 & 0 & 0 & 0 & 1 & 1 \\
1 & 0 & 1 & 0 & 1 & 0 & 0 & 0 & 1 & 0 \\
\end{array}
\right)
\]
\end{footnotesize}
where the $i$-th column corresponds to the rule $G_i$, and the $j$-th row corresponds to the critical pair $C_j$ shown in the next two pages.
\newpage
\begin{figure}[h]
{\small
\begin{align*}
  C_{1}:\ &m(m(x_1, x_2), x_3)\rrule m(x_1, m(x_2, x_3)),\quad m(m(x_4, x_5), x_6)\rrule m(x_4, m(x_5, x_6)),\quad
   m(\square , x_3),\\ &\{x_6\mapsto x_2, x_1\mapsto m(x_4, x_5)\}\\
  C_{2}:\ &i(m(x_1, x_2))\rrule m(i(x_2), i(x_1)),\quad m(m(x_3, x_4), x_5)\rrule m(x_3, m(x_4, x_5)),\quad
   i(\square ),\\ &\{x_5\mapsto x_2, x_1\mapsto m(x_3, x_4)\}\\
  C_{3}:\ &m(m(x_1, x_2), x_3)\rrule m(x_1, m(x_2, x_3)),\quad m(x_4, m(i(x_4), x_5))\rrule x_5,\quad m(\square , x_3),\\
  & \{x_2\mapsto m(i(x_1), x_5), x_4\mapsto x_1\}\\
  C_{4}:\ &m(x_1, m(i(x_1), x_2))\rrule x_2,\quad m(m(x_3, x_4), x_5)\rrule m(x_3, m(x_4, x_5)),\quad \square ,\\
  & \{x_5\mapsto m(i(m(x_3, x_4)), x_2), x_1\mapsto m(x_3, x_4)\}\\
  C_{5}:\ &m(m(x_1, x_2), x_3)\rrule m(x_1, m(x_2, x_3)),\quad m(i(x_4), m(x_4, x_5))\rrule x_5,\quad m(\square , x_3),\\
  & \{x_2\mapsto m(x_4, x_5), x_1\mapsto i(x_4)\}\\
  C_{6}:\ &m(i(x_1), m(x_1, x_2))\rrule x_2,\quad m(m(x_3, x_4), x_5)\rrule m(x_3, m(x_4, x_5)),\quad
  m(i(x_1), \square ),\\ &\{x_5\mapsto x_2, x_1\mapsto m(x_3, x_4)\}\\
  C_{7}:\ &m(m(x_1, x_2), x_3)\rrule m(x_1, m(x_2, x_3)),\quad m(i(x_4), x_4)\rrule e,\quad
   m(\square , x_3),\\& \{x_4\mapsto x_2, x_1\mapsto i(x_2)\}\\
  C_{8}:\ &m(m(x_1, x_2), x_3)\rrule m(x_1, m(x_2, x_3)),\quad m(x_4, i(x_4))\rrule e,\quad m(\square , x_3),\\
  & \{x_2\mapsto i(x_1), x_4\mapsto x_1\}\\
  C_{9}:\ &m(x_1, i(x_1))\rrule e,\quad m(m(x_2, x_3), x_4)\rrule m(x_2, m(x_3, x_4)),\quad \square ,\\
  & \{x_4\mapsto i(m(x_2, x_3)), x_1\mapsto m(x_2, x_3)\}\\
  C_{10}:\ &m(m(x_1, x_2), x_3)\rrule m(x_1, m(x_2, x_3)),\quad m(x_4, e)\rrule x_4,\quad m(\square , x_3),\quad \{x_2\mapsto e, x_4\mapsto x_1\}\\
  C_{11}:\ &m(x_1, e)\rrule x_1,\quad m(m(x_2, x_3), x_4)\rrule m(x_2, m(x_3, x_4)),\quad \square ,\quad \{x_4\mapsto e, x_1\mapsto m(x_2, x_3)\}\\
  C_{12}:\ &m(m(x_1, x_2), x_3)\rrule m(x_1, m(x_2, x_3)),\quad m(e, x_4)\rrule x_4,\quad m(\square , x_3),\quad \{x_4\mapsto x_2, x_1\mapsto e\}\\
  C_{13}:\ &i(m(x_1, x_2))\rrule m(i(x_2), i(x_1)),\quad m(e, x_3)\rrule x_3,\quad i(\square ),\quad \{x_3\mapsto x_2, x_1\mapsto e\}\\
  C_{14}:\ &m(x_1, m(i(x_1), x_2))\rrule x_2,\quad m(e, x_3)\rrule x_3,\quad \square ,\quad \{x_3\mapsto m(i(e), x_2), x_1\mapsto e\}\\
  C_{15}:\ &m(i(x_1), m(x_1, x_2))\rrule x_2,\quad m(e, x_3)\rrule x_3,\quad m(i(x_1), \square ),\quad \{x_3\mapsto x_2, x_1\mapsto e\}\\
  C_{16}:\ &m(x_1, i(x_1))\rrule e,\quad m(e, x_2)\rrule x_2,\quad \square ,\quad \{x_2\mapsto i(e), x_1\mapsto e\}\\
  C_{17}:\ &m(x_1, e)\rrule x_1,\quad m(e, x_2)\rrule x_2,\quad \square ,\quad \{x_2\mapsto e, x_1\mapsto e\}\\
  C_{18}:\ &i(m(x_1, x_2))\rrule m(i(x_2), i(x_1)),\quad m(x_3, e)\rrule x_3,\quad i(\square ),\quad \{x_2\mapsto e, x_3\mapsto x_1\}\\
  C_{19}:\ &m(x_1, m(i(x_1), x_2))\rrule x_2,\quad m(x_3, e)\rrule x_3,\quad m(x_1, \square ),\quad \{x_2\mapsto e, x_3\mapsto i(x_1)\}\\
  C_{20}:\ &m(i(x_1), m(x_1, x_2))\rrule x_2,\quad m(x_3, e)\rrule x_3,\quad m(i(x_1), \square ),\quad \{x_2\mapsto e, x_3\mapsto x_1\}\\
\end{align*}
}
  \caption{The critical pairs of the complete TRS $R$ (1)\\ ($C_j:\ l\rrule r,\  l'\rrule r',\  C,\  \sigma$ means $C_j$ is the critical pair $(r\sigma,C[r'\sigma])$.)}
\end{figure}

\pagebreak

\begin{figure}[h]
{\small
\begin{align*}
  C_{21}:\ &m(i(x_1), x_1)\rrule e,\quad m(x_2, e)\rrule x_2,\quad \square ,\quad \{x_1\mapsto e, x_2\mapsto i(e)\}\\
  C_{22}:\ &m(x_1, i(x_1))\rrule e,\quad i(m(x_2, x_3))\rrule m(i(x_3), i(x_2)),\quad m(x_1, \square ),\quad \{x_1\mapsto m(x_2, x_3)\}\\
  C_{23}:\ &i(m(x_1, x_2))\rrule m(i(x_2), i(x_1)),\quad m(x_3, i(x_3))\rrule e,\quad i(\square ),\quad \{x_2\mapsto i(x_1), x_3\mapsto x_1\}\\
  C_{24}:\ &m(x_1, m(i(x_1), x_2))\rrule x_2,\quad m(x_3, i(x_3))\rrule e,\quad m(x_1, \square ),\quad \{x_2\mapsto i(i(x_1)), x_3\mapsto i(x_1)\}\\
  C_{25}:\ &m(x_1, i(x_1))\rrule e,\quad i(i(x_2))\rrule x_2,\quad m(x_1, \square ),\quad \{x_1\mapsto i(x_2)\}\\
  C_{26}:\ &m(x_1, i(x_1))\rrule e,\quad i(e)\rrule e,\quad m(x_1, \square ),\quad \{x_1\mapsto e\}\\
  C_{27}:\ &m(i(x_1), m(x_1, x_2))\rrule x_2,\quad m(x_3, i(x_3))\rrule e,\quad m(i(x_1), \square ),\quad \{x_2\mapsto i(x_1), x_3\mapsto x_1\}\\
  C_{28}:\ &m(i(x_1), x_1)\rrule e,\quad i(m(x_2, x_3))\rrule m(i(x_3), i(x_2)),\quad m(\square , x_1),\quad \{x_1\mapsto m(x_2, x_3)\}\\
  C_{29}:\ &i(m(x_1, x_2))\rrule m(i(x_2), i(x_1)),\quad m(i(x_3), x_3)\rrule e,\quad i(\square ),\quad \{x_3\mapsto x_2, x_1\mapsto i(x_2)\}\\
  C_{30}:\ &m(x_1, m(i(x_1), x_2))\rrule x_2,\quad m(i(x_3), x_3)\rrule e,\quad m(x_1, \square ),\quad \{x_1\mapsto x_2, x_3\mapsto x_2\}\\
  C_{31}:\ &m(i(x_1), x_1)\rrule e,\quad i(i(x_2))\rrule x_2,\quad m(\square , x_1),\quad \{x_1\mapsto i(x_2)\}\\
  C_{32}:\ &m(i(x_1), x_1)\rrule e,\quad i(e)\rrule e,\quad m(\square , x_1),\quad \{x_1\mapsto e\}\\
  C_{33}:\ &m(i(x_1), m(x_1, x_2))\rrule x_2,\quad m(i(x_3), x_3)\rrule e,\quad m(i(x_1), \square ),\quad \{x_3\mapsto x_2, x_1\mapsto i(x_2)\}\\
  C_{34}:\ &m(i(x_1), m(x_1, x_2))\rrule x_2,\quad m(i(x_3), m(x_3, x_4))\rrule x_4,\quad m(i(x_1), \square ),\quad \{x_2\mapsto m(x_3, x_4), x_1\mapsto i(x_3)\}\\
  C_{35}:\ &m(i(x_1), m(x_1, x_2))\rrule x_2,\quad i(m(x_3, x_4))\rrule m(i(x_4), i(x_3)),\quad m(\square , m(x_1, x_2)),\quad \{x_1\mapsto m(x_3, x_4)\}\\
  C_{36}:\ &i(m(x_1, x_2))\rrule m(i(x_2), i(x_1)),\quad m(i(x_3), m(x_3, x_4))\rrule x_4,\quad i(\square ),\quad \{x_2\mapsto m(x_3, x_4), x_1\mapsto i(x_3)\}\\
  C_{37}:\ &m(i(x_1), m(x_1, x_2))\rrule x_2,\quad m(x_3, m(i(x_3), x_4))\rrule x_4,\quad m(i(x_1), \square ),\quad \{x_2\mapsto m(i(x_1), x_4), x_3\mapsto x_1\}\\
  C_{38}:\ &m(x_1, m(i(x_1), x_2))\rrule x_2,\quad m(i(x_3), m(x_3, x_4))\rrule x_4,\quad m(x_1, \square ),\quad \{x_2\mapsto m(x_1, x_4), x_3\mapsto x_1\}\\
  C_{39}:\ &m(i(x_1), m(x_1, x_2))\rrule x_2,\quad i(i(x_3))\rrule x_3,\quad m(\square , m(x_1, x_2)),\quad \{x_1\mapsto i(x_3)\}\\
  C_{40}:\ &m(i(x_1), m(x_1, x_2))\rrule x_2,\quad i(e)\rrule e,\quad m(\square , m(x_1, x_2)),\quad \{x_1\mapsto e\}\\
  C_{41}:\ &m(x_1, m(i(x_1), x_2))\rrule x_2,\quad i(e)\rrule e,\quad m(x_1, m(\square , x_2)),\quad \{x_1\mapsto e\}\\
  C_{42}:\ &i(i(x_1))\rrule x_1,\quad i(e)\rrule e,\quad i(\square ),\quad \{x_1\mapsto e\}\\
  C_{43}:\ &i(i(x_1))\rrule x_1,\quad i(i(x_2))\rrule x_2,\quad i(\square ),\quad \{x_1\mapsto i(x_2)\}\\
  C_{44}:\ &i(i(x_1))\rrule x_1,\quad i(m(x_2, x_3))\rrule m(i(x_3), i(x_2)),\quad i(\square ),\quad \{x_1\mapsto m(x_2, x_3)\}\\
  C_{45}:\ &m(x_1, m(i(x_1), x_2))\rrule x_2,\quad i(i(x_3))\rrule x_3,\quad m(x_1, m(\square , x_2)),\quad \{x_1\mapsto i(x_3)\}\\
  C_{46}:\ &m(x_1, m(i(x_1), x_2))\rrule x_2,\quad m(x_3, m(i(x_3), x_4))\rrule x_4,\quad m(x_1, \square ),\\
  & \{x_2\mapsto m(i(i(x_1)), x_4), x_3\mapsto i(x_1)\}\\
  C_{47}:\ &m(x_1, m(i(x_1), x_2))\rrule x_2,\quad i(m(x_3, x_4))\rrule m(i(x_4), i(x_3)),\quad m(x_1, m(\square , x_2)),\quad \{x_1\mapsto m(x_3, x_4)\}\\
  C_{48}:\ &i(m(x_1, x_2))\rrule m(i(x_2), i(x_1)),\quad m(x_3, m(i(x_3), x_4))\rrule x_4,\quad i(\square ),\quad \{x_2\mapsto m(i(x_1), x_4), x_3\mapsto x_1\}
\end{align*}
}
  \caption{The critical pairs of the complete TRS $R$ (2)\\ ($C_j:\ l\rrule r,\  l'\rrule r',\  C,\  \sigma$ means $C_j$ is the critical pair $(r\sigma,C[r'\sigma])$.)}
\end{figure}

%% file: result.tex
\label{section:experiment}
We present our experimental data in Table 1 and Table 2.
The data set of complete TRSs is taken from experimental results of MKBtt \cite{SWKM08}, which include benchmark problems \cite{sk90},\cite{ws06},\cite{wsw}.
The column headed ``degree'' shows the degree of the TRS,
the column $\#R_\textit{before}$ the number of rules,
the column $\#R_\textit{after}$ the number of rules after completion,
the column $s(H_2)$ Malbos-Mimram's lower bound,
and the column $\#R_\textit{after}-e(R)$ our lower bound.
The table is also available at \url{https://mir-ikbch.github.io/homtrs/experiment/result.html}
which has links to TRS files.
\begin{table}[h]
\caption{Malbos-Mimram's and our lower bounds (1)}
\begin{tabular}{lrrrrr}
name & degree & 
\(\#R_{\mathit{before}}\) &
\(\#R_{\mathit{after}}\) &
\(s(H_2)\) &
\(\#R_{\mathit{after}}-e(R)\)
\\
\hline
\texttt{ASK93\_1} &
  0 & 2 &
  2 &
  0 & 2
\\
\texttt{ASK93\_6} &
  0 & 11 &
  11 &
  0 & 9
\\
\texttt{BD94\_collapse} &
  1 & 5 &
  5 &
  -- & --
\\
\texttt{BD94\_peano} &
  1 & 4 &
  4 &
  -- & --
\\
\texttt{BD94\_sqrt} &
  2 & 3 &
  4 &
  0 & 3
\\
\texttt{BGK94\_D08} &
  2 & 6 &
  21 &
  2 & 5
\\
\texttt{BGK94\_D10} &
  2 & 6 &
  21 &
  1 & 4
\\
\texttt{BGK94\_D12} &
  2 & 6 &
  20 &
  2 & 5
\\
\texttt{BGK94\_D16} &
  2 & 6 &
  20 &
  2 & 5
\\
\texttt{BH96\_fac8\_theory} &
  1 & 6 &
  6 &
  -- & --
\\
\texttt{Chr89\_A2} &
  2 & 5 &
  18 &
  0 & 4
\\
\texttt{Chr89\_A3} &
  2 & 7 &
  16 &
  0 & 6
\\
\texttt{KK99\_linear\_assoc} &
  0 & 2 &
  2 &
  0 & 1
\\
\texttt{LS94\_G0} &
  2 & 8 &
  13 &
  1 & 4
\\
\texttt{Les83\_fib} &
  1 & 9 &
  9 &
  -- & --
\\
\texttt{Les83\_subset} &
  1 & 12 &
  12 &
  -- & --
\\
\texttt{OKW95\_dt1\_theory} &
  1 & 11 &
  11 &
  -- & --
\\
\texttt{SK90\_3.01} &
  2 & 4 &
  11 &
  0 & 3
\\
\texttt{SK90\_3.02} &
  0 & 3 &
  3 &
  1 & 2
\\
\texttt{SK90\_3.03} &
  2 & 5 &
  11 &
  0 & 3
\\
\texttt{SK90\_3.04} &
  1 & 4 &
  8 &
  -- & --
\\
\texttt{SK90\_3.05} &
  1 & 4 &
  13 &
  -- & --
\\
\texttt{SK90\_3.06} &
  1 & 5 &
  12 &
  -- & --
\\
\texttt{SK90\_3.07} &
  1 & 5 &
  15 &
  -- & --
\\
\texttt{SK90\_3.08} &
  2 & 5 &
  4 &
  0 & 2
\\
\texttt{SK90\_3.10} &
  2 & 4 &
  8 &
  0 & 3
\\
\texttt{SK90\_3.11} &
  0 & 4 &
  3 &
  0 & 3
\\
\texttt{SK90\_3.12} &
  2 & 4 &
  9 &
  0 & 2
\\
\texttt{SK90\_3.13} &
  0 & 6 &
  6 &
  0 & 3
\\
\texttt{SK90\_3.14} &
  0 & 7 &
  8 &
  1 & 5
\\
\texttt{SK90\_3.15} &
  2 & 8 &
  7 &
  1 & 4
\\
\texttt{SK90\_3.16} &
  1 & 4 &
  4 &
  -- & --
\\
\texttt{SK90\_3.17} &
  1 & 3 &
  5 &
  -- & --
\end{tabular}
\end{table}

%% file: result_cont.tex
\begin{table}[h]
\caption{Malbos-Mimram's and our lower bounds (2)}
\begin{tabular}{lrrrrr}
name & degree & 
\(\#R_{\mathit{before}}\) &
\(\#R_{\mathit{after}}\) &
\(s(H_2)\) &
\(\#R_{\mathit{after}}-e(R)\)
\\
\hline
\\
\texttt{SK90\_3.18} &
  0 & 5 &
  6 &
  2 & 4
\\
\texttt{SK90\_3.19} &
  0 & 9 &
  7 &
  1 & 4
\\
\texttt{SK90\_3.20} &
  1 & 10 &
  11 &
  -- & --
\\
\texttt{SK90\_3.21} &
  1 & 9 &
  4 &
  -- & --
\\
\texttt{SK90\_3.23} &
  0 & 4 &
  8 &
  1 & 4
\\
\texttt{SK90\_3.24} &
  0 & 3 &
  2 &
  0 & 2
\\
\texttt{SK90\_3.25} &
  0 & 1 &
  2 &
  0 & 1
\\
\texttt{SK90\_3.27} &
  0 & 8 &
  3 &
  0 & 3
\\
\texttt{SK90\_3.28} &
  0 & 9 &
  18 &
  0 & 6
\\
\texttt{SK90\_3.29} &
  0 & 7 &
  8 &
  2 & 7
\\
\texttt{SK90\_3.30} &
  1 & 3 &
  3 &
  -- & --
\\
\texttt{SK90\_3.31} &
  1 & 3 &
  3 &
  -- & --
\\
\texttt{SK90\_3.32} &
  1 & 3 &
  2 &
  -- & --
\\
\texttt{SK90\_3.33} &
  0 & 3 &
  3 &
  0 & 2
\\
\texttt{TPTP-BOO027-1\_theory} &
  1 & 5 &
  5 &
  -- & --
\\
\texttt{TPTP-COL053-1\_theory} &
  0 & 1 &
  1 &
  0 & 1
\\
\texttt{TPTP-COL056-1\_theory} &
  0 & 3 &
  3 &
  0 & 3
\\
\texttt{TPTP-COL060-1\_theory} &
  0 & 2 &
  2 &
  0 & 2
\\
\texttt{TPTP-COL085-1\_theory} &
  0 & 1 &
  1 &
  0 & 1
\\
\texttt{TPTP-GRP010-4\_theory} &
  2 & 4 &
  11 &
  1 & 3
\\
\texttt{TPTP-GRP011-4\_theory} &
  2 & 4 &
  11 &
  1 & 3
\\
\texttt{TPTP-GRP012-4\_theory} &
  2 & 4 &
  10 &
  0 & 2
\\
\texttt{slothrop\_ackermann} &
  1 & 3 &
  3 &
  -- & --
\\
\texttt{slothrop\_cge} &
  2 & 6 &
  20 &
  0 & 4
\\
\texttt{slothrop\_cge3} &
  2 & 9 &
  28 &
  0 & 5
\\
\texttt{slothrop\_endo} &
  2 & 4 &
  14 &
  0 & 3
\\
\texttt{slothrop\_equiv\_proofs} &
  1 & 12 &
  23 &
  -- & --
\\
\texttt{slothrop\_fgh} &
  1 & 4 &
  3 &
  -- & --
\\
\texttt{slothrop\_groups} &
  2 & 3 &
  10 &
  0 & 2
\\
\texttt{slothrop\_groups\_conj} &
  2 & 5 &
  10 &
  0 & 2
\\
\texttt{slothrop\_hard} &
  0 & 2 &
  2 &
  1 & 2
\end{tabular}
\end{table}